\documentclass[a4paper, 12pt]{amsart}%
\usepackage{amsmath}
\usepackage{graphicx}
\usepackage{braket}%
\usepackage{amsfonts}%
\usepackage{amssymb}

\newtheorem{theorem}{Theorem}[section]
\newtheorem{proposition}[theorem]{Proposition}
\newtheorem{lemma}[theorem]{Lemma}
\newtheorem{corollary}[theorem]{Corollary}
\theoremstyle{definition}
\newtheorem{definition}[theorem]{Definition}

\makeatletter
\@namedef{subjclassname@2020}{	\textup{2020} Mathematics Subject Classification}
\makeatother

\begin{document}
\date{2022-10-21}
\title[Quantum Wasserstein distance of order 1]{Quantum Wasserstein distance of order 1 between channels}
\author{Rocco Duvenhage and Mathumo Mapaya}
\address{Department of Physics\\
University of Pretoria\\
Pretoria 0002\\
South Africa}
\email{rocco.duvenhage@up.ac.za}

\begin{abstract}
We set up a general theory for a quantum Wasserstein distance of order 1 in an
operator algebraic framework, extending recent work in finite dimensions. In
addition, this theory applies not only to states, but also to channels, giving
a metric on the set of channels from one composite system to another. The
additivity and stability properties of this metric are studied.

\emph{Keywords}: quantum optimal transport; quantum Wasserstein distance of
order 1; quantum channels

\end{abstract}
\maketitle
\tableofcontents

\section{Introduction}

This paper is devoted to devising a distance between quantum channels which is
natural in the context of composite quantum systems. Our approach is motivated
by a recently introduced quantum Wasserstein distance of order 1 between
states of a composite system \cite{DMTL}, as a natural and very effective
distinguishability measure with a number of desirable properties.

The method that \cite{DMTL} uses to set up their quantum Wasserstein distance
of order 1, will be referred to as the DMTL approach, and the resulting metric
can be viewed as a quantum version of Ornstein's $\bar{d}$ distance
\cite{Orn}. It is starts with the concept of neighbouring states, a form of
which was introduced earlier in \cite{AR}. We use some of the core ideas of
the DMTL approach, but adapt them to channels. Although the term Wasserstein
distance is typically used for states (or probability distributions in the
classical case), we nevertheless continue to use the term in the case of
channels as well.

We work in a more general infinite dimensional setting, expressed in terms of
C*-algebras as a generalization of the matrix algebras $M_{d}$. Mathematically
speaking, we obtain a metric on the set of all channels from one composite
system to another, where both consist of $n$ systems. The case of states is
obtained when the channels are taken to map to the complex numbers. In
particular, the DMTL approach is recovered when restricting to finite dimensions.

Broadly speaking, our approach is to extend the idea of neighbouring states to
neighbouring channels. This is done through the reduction of channels. Given a
channel from one composite system to another, choose the $j$'th system in
each, and reduce the channel to the remaining systems. If two channels give
the same channel after this reduction for some $j$ is applied to both, we
consider them to be neighbouring channels. I.e., two channels are called
neighbouring if they coincide after the removal of two corresponding systems
from the two composite systems respectively. The neighbouring channels are
then used to build a unit ball leading to a norm on a certain space of linear
maps and ultimately to a distance between any two channels. More precisely, we
obtain a metric on the set of channels from the one composite system to the
other. All this is illustrated concretely in finite dimensions in the next section.

One of our intentions is to highlight the general mathematical structure
behind the quantum Wasserstein distance of order 1, through an abstract
approach. This involves first setting up the theory in general vector spaces,
and subsequently algebras, without reference to any positivity conditions on
the maps between these spaces. The C*-algebraic framework, with the maps taken
as channels, is obtained as a special case of the abstract setup. This
clarifies the overall structure of our approach, and potentially allows for
cases other than C*-algebras and channels between them.

We nevertheless also emphasize the concrete case in finite dimensions,
particularly in the next section, to clarify the basic approach and ideas. The
C*-algebraic framework is a generalization of the finite dimensional case, and
very natural and relevant from the view of quantum physics. The paper is
written in such a way that the main thread and results in the C*-algebraic
case can be followed without going through the more abstract approach
mentioned above. The finite dimensional situation is in turn obtained as a
simple special case of the C*-algebraic framework. The proofs of our results
do depend on the abstract development, though.

The next section gives an outline of our approach in finite dimensions,
expressed in a form very close to that of the DMTL approach. It can be viewed
as an extension of this introduction, explaining some of the goals and
motivation for this paper, but it also serves as an overview for readers whose
main interest is the finite dimensional case. Section \ref{AfdAgtergrond}
reviews basic background related to norms and unit balls, and takes an initial
step in developing the theory behind the Wasserstein distance of order 1
between channels. Sections \ref{AfdGepunt} and \ref{AfdAlg} develop the
abstract theory to obtain our Wasserstein distance of order 1. The reader can
in fact skip these two sections upon initial reading, and go directly to
Section \ref{AfdC*}. There the definition and main results leading to the
Wasserstein distance of order 1 between channels in the C*-algebraic
framework, are presented with no reference to Sections \ref{AfdGepunt} and
\ref{AfdAlg}, although the proofs depend on those two sections. Note that
we'll usually simply say ``Wasserstein distance'' instead of ``quantum
Wasserstein distance'', since our theory contains the special case of abelian C*-algebras.

The paper then proceeds to the behaviour of the Wasserstein distance of order
1 with respect to subsystems of the composite systems, i.e., when we consider
smaller composite systems consisting of a subset of the original systems. The
main result here is that this Wasserstein distance is additive over tensor
products of channels between such subsystems, with stability as a special
case. Again this result is first approached abstractly in general vector
spaces, before the C*-algebraic case is presented. As before, the reader can
page directly to Subsection \ref{OndAfdC*ad} after Section \ref{AfdC*}, to see
the C*-algebraic results with a minimum of direct reference to the abstract
theory, although the proofs again rely on the latter.

To conclude this introduction, we briefly discuss previous work.

The study of quantum channels, including mathematical techniques to analyze
them theoretically, remains undeniably important. Refer for example to
\cite{CGLM, HG, KCSC}\ for an overview of a variety of aspects of quantum
channels and their significance. In particular, there have been other
approaches to distances between channels, in relation to channel
discrimination. For example, the diamond norm, also called the completely
bounded trace norm; see \cite{Kit, AKitN} for early work in connection to
quantum computation, and \cite[Section 3.3]{Wat} for a more general finite
dimensional overview in the context of quantum information. Refer to
\cite{GLN} for a broader perspective on distances between channels.

On the other hand, there has also been much effort to obtain quantum (or
noncommutative) Wasserstein distances between states. Papers obtaining
Wasserstein distances (with a focus on order 2) which are actual metrics on
sets of quantum states, include \cite{BV, CM1, CM2, CM3, CGGT, CGT, CEFZ, D1,
H1, W}. These papers follow different approaches from this paper and
\cite{DMTL}, and to a large extent from one another. Broadly speaking there
have been two main approaches, in analogy to the classical case: a coupling
(or transport plan) approach \cite{BV, CEFZ, D1} and a dynamical approach
\cite{CM1, CM2, CM3, CGGT, CGT, H1, W}. The book \cite{V1} includes a nice
introduction to the former, while \cite{BB} is the origin of the latter. The
DMTL approach and this paper, can roughly be classified as part of the
coupling approach. Other papers using the coupling approach, though not
obtaining all the usual metric properties, include \cite{dePT, GMP, GP1, GP2}.
Papers treating various other approaches than \cite{DMTL} to noncommutative or
quantum Wasserstein distances of order 1, include \cite{Ag, CM3, CGNT, GJL,
NG, RD, RCLO}. Some of these use a dual formulation of the coupling approach,
which is closely related to Connes' spectral distance in noncommutative
geometry, first introduced in \cite{Con89}, and studied further in \cite{Con,
Con96, D'AM, Rief99}, among others.

It seems that a theory of Wasserstein distances, by any approach, has not
before been extended from states to channels in the literature. This is of
course aside from attempting to apply the former directly to the latter via
the Choi-Jamio{\l}kowski duality between states and channels. This duality
holds under quite general conditions and could be applied beyond just the
finite dimensional case; see in particular \cite[Section 3]{DS2}. However, the
intention of this paper is rather to systematically rebuild such a theory for
the channels, which is not equivalent to merely translating via the
Choi-Jamio{\l}kowski duality.

\section{Outline in finite dimensions\label{AfdBuitL}}

In this section we briefly outline the paper's approach and core results in
finite dimensions, to make the basic ideas as accessible as possible. To
clarify how the DMTL approach is being extended, we use a formulation
analogous to theirs.

To do this, we in particular make use of the Choi-Jamio{\l}kowski duality
between channels and states. To avoid confusion, however, it is important to
note from the outset that the latter states will not be employed as the states
in the DMTL approach. That is to say, we do not simply apply their results to
states representing channels via the Choi-Jamio{\l}kowski duality. Rather, the
DMTL approach is systematically adapted directly to channels, while the
representation of channels via Choi-Jamio{\l}kowski duality simply emphasizes
the connection to the DMTL approach, as will be seen below.

In the general theory presented in later sections, on the other hand, it will
be more convenient to avoid the Choi-Jamio{\l}kowski duality, instead working
with the channels directly. The finite dimensional setup discussed in this
section emerges as a special case of that general theory.

\subsection{A representation of channels}

Consider a channel $\mathcal{E}$ from the density matrices of a system
$\mathcal{B}$ to that of a system $\mathcal{A}$. Fix an orthonormal basis
$\ket{1^\mathcal{B}}  ,...,\ket{r^\mathcal{B}}  $ for the Hilbert space
$H_{\mathcal{B}}$ of the system $\mathcal{B}$, and define a density matrix
$\kappa_{\mathcal{E}}$ of the composite system $\mathcal{BA}$ by
\begin{equation}
\kappa_{\mathcal{E}}=\frac{1}{r}\sum_{i=1}^{r}\sum_{j=1}^{r}\ket{i^\mathcal{B}%
}  \bra{j^\mathcal{B}}  \otimes\mathcal{E}(\ket{i^\mathcal{B}}
\bra{j^\mathcal{B}}  ). \label{CJ}%
\end{equation}
This density matrix reduces to the maximally mixed state of $\mathcal{B}$. The
Choi-Jamio{\l}kowski duality \cite{Choi} tells us that there is a one-to-one
correspondence between the set of such channels $\mathcal{E}$, and the set of
density matrices of $\mathcal{BA}$ reducing to the maximally mixed state of
$\mathcal{B}$. (Also see \cite{deP, J}, but Choi's approach \cite{Choi} forms
the basis for our use of the Choi-Jamio{\l}kowski duality in this section.)

Represent the observable algebra of $\mathcal{B}$ in terms of the given basis
as the matrix algebra $B=M_{r}$, and similarly use $A=M_{q}$ for $\mathcal{A}%
$. To keep the distinction between states and observables conceptually and
notationally clear, we also write $S_{\mathcal{A}}=M_{q}$ for the matrix
algebra containing the density matrices of $\mathcal{A}$, and similarly
$S_{\mathcal{B}}=M_{r}$ for $\mathcal{B}$. In particular, the channel
$\mathcal{E}$ is then a completely positive trace preserving linear map%
\[
\mathcal{E}:S_{\mathcal{B}}\rightarrow S_{\mathcal{A}}%
\]
from $S_{\mathcal{B}}$ to $S_{\mathcal{A}}$. This channel's dual
representation in terms of the observable algebras is the unital completely
positive linear map%
\[
E:A\rightarrow B
\]
defined via%
\[
\operatorname*{Tr}[E(a)b]=\operatorname*{Tr}[a\mathcal{E}(b)]
\]
for all $a\in A$ and $b\in S_{\mathcal{B}}$.

Consider the set $K(A,B)$ of channels $E:A\rightarrow B$, and the set
\[
C(\mathcal{BA})=\{\delta\in B\otimes S_{\mathcal{A}}:\delta\geq0\text{ \ and
\ }\operatorname*{Tr}\nolimits_{\mathcal{A}}\delta=1_{B}\},
\]
(not $\delta\in S_{\mathcal{B}}\otimes S_{\mathcal{A}}$; see below) where
$\operatorname*{Tr}\nolimits_{\mathcal{A}}$ denotes the partial trace over the
system $\mathcal{A}$, and $1_{B}$ is the identity matrix in $B=M_{r}$.
Implicit to the Choi-Jamio{\l}kowski duality, and fairly straightforward to
extract from it, we have a one-to-one correspondence between $K(A,B)$ and
$C(\mathcal{BA})$ given by
\begin{equation}
E(a)=\operatorname*{Tr}\nolimits_{\mathcal{A}}[\delta(1_{B}\otimes a)].
\label{E&delta}%
\end{equation}
Denoting the dual of $E$ by $\mathcal{E}:S_{\mathcal{B}}\rightarrow
S_{\mathcal{A}}$, we note that in terms of the Choi-Jamio{\l}kowski duality as
stated above, one has
\[
\delta=r\kappa_{\mathcal{E}}^{\text{T}_{\mathcal{B}}}%
\]
in this one-to-one correspondence, where T$_{\mathcal{A}}$ is the partial
transposition over $\mathcal{B}$ with respect to the chosen basis for
$H_{\mathcal{B}}$. Because of the usual interpretation of the Choi-Jamio{\l
}kowski duality, one might be tempted to rather view the elements of
$C(\mathcal{BA})$ as $\delta\in S_{\mathcal{B}}\otimes S_{\mathcal{A}}$ such
that $\delta/r$ is a density matrix of the composite system $\mathcal{BA}$
reducing to the maximally mixed state of $\mathcal{B}$, but strictly speaking
$E(a)$ as given by (\ref{E&delta}) would then be in $S_{\mathcal{B}}$ instead
of $B$. When $r>1$, it is therefore in fact conceptually better if we do not
view $\delta/r$ as a density matrix representing some state. On the other
hand, in the special case where $\mathcal{B}$ is a trivial system, i.e., $r=1$
and $B=\mathbb{C}$, the set $C(\mathcal{BA})$ is precisely all the density
matrices of $\mathcal{A}$, and we simply recover the usual representation of
expectation values of a state, $E(a)=\operatorname*{Tr}(\delta a)$, in terms
of the density matrix $\delta$.

\subsection{The $W_{1}$ norm}

Using this representation of channels as elements of $C(\mathcal{BA})$, we can
formulate an extension of the Wasserstein distance of order 1 between states
in the DMTL approach, to channels. This is a distance between channels acting
from one composite system, $\mathcal{A}$, to another, $\mathcal{B}$. Here we
assume that
\[
A=A_{1}\otimes...\otimes A_{n}\text{ \ and \ }B=B_{1}\otimes...\otimes B_{n}%
\]
with $A_{j}=M_{q_{j}}$ and $B_{j}=M_{r_{j}}$. The latter are simply the
observable algebras of systems $\mathcal{A}_{j}$ and $\mathcal{B}_{j}$
respectively. One then defines the real vector space%
\begin{equation}
\mathcal{O}=\{X\in\operatorname*{span}\nolimits_{\mathbb{R}}C(\mathcal{BA}%
):\operatorname*{Tr}\nolimits_{\mathcal{A}}X=0\}, \label{spesO}%
\end{equation}
where $\operatorname*{span}\nolimits_{\mathbb{R}}$ refers to finite linear
combinations with real coefficients. The goal is to define a certain norm
$\left\|  \cdot\right\|  _{W_{1}}$ on $\mathcal{O}$, which when applied to
differences $\delta-\varepsilon$ of elements of $C(\mathcal{BA})$, will in
turn define a metric on $C(\mathcal{BA})$. This metric will be the Wasserstein
distance of order 1 on $C(\mathcal{BA})$, or equivalently on $K(A,B)$ via the
one-to-one correspondence above, extending the construction in the DMTL
approach to channels.

The definition of the norm $\left\|  \cdot\right\|  _{W_{1}}$ entails
extending the idea of neighbouring states in the DMTL approach to channels. We
write%
\[
A_{\widehat{j}}=A_{1}\otimes...\widehat{A}_{j}...\otimes A_{n}\text{ \ and
\ }B_{\widehat{j}}=B_{1}\otimes...\widehat{B}_{j}...\otimes B_{n},
\]
i.e., these are $A$ and $B$ with $A_{j}$ and $B_{j}$ respectively left out of
the tensor products. Similarly we set%
\[
A_{\leq j}=A_{1}\otimes...\otimes A_{j}\text{ \ and \ }A_{\geq j}=A_{j}%
\otimes...\otimes A_{n},
\]
and likewise for $B_{\leq j}$ and $B_{\geq j}$. In terms of the following
notation (which is chosen to fit in with that of later sections),%
\[
\nu_{j}=\frac{1}{r_{j}}\operatorname*{Tr}%
\]
where this $\operatorname*{Tr}$ is the usual trace on $B_{j}=M_{r_{j}}$, we
can then reduce a channel $E:A\rightarrow B$ to the channel%
\[
E_{\widehat{j}}=(\operatorname*{id}\nolimits_{B_{\leq j-1}}\otimes\nu
_{j}\otimes\operatorname*{id}\nolimits_{B_{\geq j+1}})\circ E|_{A_{\widehat
{j}}}%
\]
from $A_{\widehat{j}}$ to $B_{\widehat{j}}$, where $\operatorname*{id}$
denotes the identity map on the indicated algebra, and with $E|_{A_{\widehat
{j}}}$ defined via
\[
E|_{A_{\widehat{j}}}(a_{1}\otimes...\widehat{a}_{j}...\otimes a_{n}%
)=E(a_{1}\otimes...\otimes a_{j-1}\otimes1_{A_{j}}\otimes a_{j+1}%
\otimes...\otimes a_{n}).
\]
I.e., we restrict $E$ to $A_{\widehat{j}}$, and evaluate the ``partial
expectation'' of the result over $B_{j}$. Note that via the one-to-one
correspondence given by (\ref{E&delta}), this reduction of $E$ is equivalent
to the reduction
\[
\delta_{\widehat{j}}=(\operatorname*{id}\nolimits_{B_{\leq j-1}}\otimes\nu
_{j}\otimes\operatorname*{id}\nolimits_{B_{\geq j+1}})\otimes
\operatorname*{Tr}\nolimits_{\mathcal{A}_{j}}%
\]
of the corresponding $\delta\in C(\mathcal{BA})$, where $\operatorname*{Tr}%
\nolimits_{\mathcal{A}_{j}}$ denotes the partial trace on $S_{\mathcal{A}%
}=M_{q_{1}}\otimes...\otimes M_{q_{n}}$ over $S_{\mathcal{A}_{j}}=M_{q_{j}}$.
In terms of this notation, we view $\delta,\varepsilon\in C(\mathcal{BA})$ as
representing neighbouring channels when $\delta_{\widehat{j}}=\varepsilon
_{\widehat{j}}$ for some $j$. We define
\begin{equation}
\mathcal{N}_{j}=\{\delta-\varepsilon:\delta,\varepsilon\in C(\mathcal{BA}%
)\text{ with }\delta_{\widehat{j}}=\varepsilon_{\widehat{j}}\} \label{spesN_j}%
\end{equation}
and
\[
\mathcal{N}=\bigcup_{j=1}^{n}\mathcal{N}_{j}.
\]
We then define $\left\|  \cdot\right\|  _{W_{1}}$ as the norm on $\mathcal{O}$
which has the convex hull
\[
\mathcal{C}=\operatorname*{conv}\mathcal{N}%
\]
of $\mathcal{N}$ as its unit ball. I.e.,%
\begin{equation}
\left\|  X\right\|  _{W_{1}}=\inf\{t\geq0:X\in t\mathcal{C}\}
\label{spesW-norm}%
\end{equation}
where $t\mathcal{C}=\{tX:X\in\mathcal{C}\}$ for any real number $t$. We write
this as an infimum, rather than a minimum, to conform to our more general
approach later on. The norm $\left\|  \cdot\right\|  _{W_{1}}$ will be
referred to as the $W_{1}$\emph{ norm}. To prove that this is a norm on
$\mathcal{O}$ of course requires some work, which will be done in a more
general context in the sequel.

\subsection{Wasserstein distance of order 1 between channels}

Given this norm, we can define a metric $W_{1}$ on $K(A,B)$ via%
\[
W_{1}(E_{\delta},E_{\varepsilon})=\left\|  \delta-\varepsilon\right\|
_{W_{1}}%
\]
with $E_{\delta}$ denoting the channel $E$ corresponding to $\delta\in
C(\mathcal{BA})$ in (\ref{E&delta}). This metric $W_{1}$ is the generalization
of the Wasserstein distance of order 1 between states to the case of channels.
We consequently refer to it as the \emph{Wasserstein distance of order 1} on
$K(A,B)$.

The rough intuition behind this metric follows from $\delta$ and $\varepsilon$
in the definition of $\mathcal{N}_{j}$ above being neighbouring channels. This
condition tells us that $\varepsilon$ and $\delta$ coincide when \ reduced to
$A_{\widehat{j}}\rightarrow B_{\widehat{j}}$ for some $j$, i.e., with one
system, $\mathcal{A}_{j}$ and $\mathcal{B}_{j}$ respectively, removed from
each of the composite systems $\mathcal{A}$ and $\mathcal{B}$. In this way the
``local differences'' between two channels $E_{\delta}$ and $E_{\varepsilon}$
are picked up by $W_{1}$, where ``local'' here is simply in relation to the
systems $\mathcal{A}_{1},...,\mathcal{A}_{n}$ and $\mathcal{B}_{1}%
,...,\mathcal{B}_{n}$ composing $\mathcal{A}$ and $\mathcal{B}$. A typical
case is the dynamics of an open composite system $\mathcal{A}$, where we take
$\mathcal{B}_{j}=\mathcal{A}_{j}$ for all $j$. We then expect $W_{1}$ to
naturally take into account the differences between two dynamical processes in
the individual systems $\mathcal{A}_{j}$.

A basic property of $W_{1}$ is additivity with respect to tensor products. If
we partition the set $[n]=\{1,...,n\}$ into $m$ non-empty and sequential
parts, i.e.,%
\begin{align*}
P(1)  &  =1,...,n_{1}\\
P(2)  &  =n_{1}+1,...,n_{2}\\
&  \vdots\\
P(m)  &  =n_{m-1}+1,...,n,
\end{align*}
then we can consider the subsystems of $\mathcal{A}$ and $\mathcal{B}$ with
observable algebras
\[
A_{P(k)}=\bigotimes_{j\in P(k)}A_{j}\text{ \ and \ }B_{P(k)}=\bigotimes_{j\in
P(k)}B_{j}%
\]
respectively. For any channels $D_{k},E_{k}:$ $A_{P(k)}\rightarrow B_{P(k)}$
we then have%
\[
W_{1}(D_{1}\otimes...\otimes D_{m},E_{1}\otimes...\otimes E_{m})=\sum
_{k=1}^{m}W_{1}(D_{k},E_{k}),
\]
with a resulting stability property when $D_{j}=E_{j}$ for some of the $j$'s.
This can be refined by dropping the assumption that the partition is
sequential, but the form above is for the moment notationally clearer. Keep in
mind that $W_{1}$ on $K(A_{P(k)},B_{P(k)})$ is of course defined by the same
procedure as for $K(A,B)$.

When $\mathcal{B}_{j}=\mathbb{C}$, we recover the case of states on
$\mathcal{A}$, and indeed our $W_{1}$ above then specializes to the
Wasserstein distance of order 1 in the DMTL approach, with $\delta
,\varepsilon\in C(\mathcal{BA})$ becoming density matrices $\rho$ and $\sigma$
of the composite system $\mathcal{A}$. They specifically considered the case
$q_{1}=...=q_{n}=d$.

In the subsequent general theory, the formulation will be in the Heisenberg
picture $E:A\rightarrow B$ from the outset. The Choi-Jamio{\l}kowski duality
will also be side-stepped, with the formulation expressed directly in terms of
the channels themselves. In connection to this, note that the basic condition
$\operatorname*{Tr}\nolimits_{\mathcal{A}}X=0$ in (\ref{spesO}) can
equivalently be expressed as $\lambda(1_{A})=0$, with $\lambda:A\rightarrow B$
defined by%
\[
\lambda(a)=\operatorname*{Tr}\nolimits_{\mathcal{A}}[X(1_{B}\otimes a)]
\]
in terms of the given $X$. The latter extends the formula in (\ref{E&delta}).
The condition $\lambda(1_{A})=0$ relates to the unitality of channels, namely
$E(1_{A})=1_{B}$, as will be seen in abstract form in Sections \ref{AfdGepunt}
and \ref{AfdAlg}.

The reader may now turn directly to Section \ref{AfdC*} and\textbf{
}\ref{OndAfdC*ad} to see the general C*-algebraic version of this section.
However, the proofs that the Wasserstein distance of order 1 treated there is
indeed a metric and satisfies additivity, rely on Sections \ref{AfdGepunt} and
\ref{AfdAlg} as well as the rest of \ref{AfdSubStelsels}.

\section{Elementary background\label{AfdAgtergrond}}

This section reviews and sets up some basic notions and elementary results to
be used later on. This is for easy reference and since references providing
this material in exactly the form we need appear to be scarce. Presentations
of this material, though in somewhat different forms than what we use, can be
found in \cite[Section I.1]{BD} and \cite[Chapters 4 and 5]{NB}. The results
in this section are stated for real vector spaces, since this is exactly what
we need later on, as can already be seen from (\ref{spesO}). Note that we do
not assume any of the vector spaces to be finite dimensional.

The basic result is how a norm arises from a given set which is to serve as a
unit ball for the norm.

\begin{definition}
\label{eenhBal}Consider a subset $\mathcal{C}$ of a real vector space
$\mathcal{X}$.

(1) The set $\mathcal{C}$ is called \emph{absorbing (for} $\mathcal{X)}$ if
for every $x\in\mathcal{X}$ there is a number $t\geq0$ such that $x\in
t\mathcal{C}$, where $t\mathcal{C}=\{tx:x\in\mathcal{C}\}$. The \emph{gauge}
$\left\|  \cdot\right\|  _{\mathcal{C}}:X\rightarrow\mathbb{R}_{+}$ of such an
absorbing set $\mathcal{C}$ is the function on $\mathcal{X}$ defined by
\[
\left\|  x\right\|  _{\mathcal{C}}:=\inf\{t\geq0:x\in{t}\mathcal{C}\}
\]
for every $x\in\mathcal{X}$, where $\mathbb{R}_{+}:=\{t\in\mathbb{C}:t\geq0\}$.

(2) The set $\mathcal{C}$ is called \emph{symmetric} if $\mathcal{C}%
=-\mathcal{C}$.

(3) The set $\mathcal{C}$ is called \emph{ray-wise (or radially) bounded (in}
$\mathcal{X)}$ if for every non-zero $x\in\mathcal{X}$ there is a number
$s_{0}>0$ such that $sx\notin\mathcal{C}$ for all $s>s_{0}$.

(4) If $\mathcal{C}$ is absorbing, convex and symmetric, then we call it a
\emph{semi unit ball for} $\mathcal{X}$. If in addition $\mathcal{C}$ is
ray-wise bounded, then we call it a \emph{unit ball for} $\mathcal{X}$.
\end{definition}

The following simple result is the method by which we intend to define our
norms, as was seen in the special case (\ref{spesW-norm}).

\begin{proposition}
\label{normUitBal}Let $\mathcal{C}$ be a semi unit ball for a real vector
space $\mathcal{X}$. \smallskip Then its gauge $\left\|  \cdot\right\|
_{\mathcal{C}}$ is a seminorm on $\mathcal{X}$. If $\mathcal{C}$ is a unit
ball for $\mathcal{X}$, then $\left\|  \cdot\right\|  _{\mathcal{C}}$ is a
norm on $\mathcal{X}$.
\end{proposition}

Results of this type are well known and straightforward to prove. That
$\mathcal{C}$ is absorbing, ensures that $\left\|  \cdot\right\|
_{\mathcal{C}}$ is well defined and finite at every point of $\mathcal{X}$,
convexity implies the triangle inequality, symmetry leads to $\left\|  \alpha
x\right\|  _{\mathcal{C}}=|\alpha|\left\|  x\right\|  _{\mathcal{C}}$ for all
$\alpha\in\mathbb{R}$ (these three properties making $\left\|  \cdot\right\|
_{\mathcal{C}}$ a seminorm), while ray-wise boundedness enforces $x=0$ when
$\left\|  x\right\|  _{\mathcal{C}}=0$. The harder problem is to define an
appropriate set $\mathcal{C}$ in an operator algebraic context (generalizing
the previous section) and prove that it is indeed a unit ball, which is
exactly the focus in the next three sections. In the meantime we note an
elementary result which will be relevant in the next section when proving the
absorbing property.

\begin{lemma}
\label{abstrakteAbs}Let $\mathcal{X}$ be a real vector space. Consider any
symmetric subset $\mathcal{N}$ of $\mathcal{X}$ such that $\mathcal{X}$
$=\operatorname*{span}\mathcal{N}$. Then its convex hull $\mathcal{C}%
=\operatorname*{conv}\mathcal{N}$ is absorbing for $\mathcal{X}$.
\end{lemma}

\begin{proof}
For any non-zero $x\in\mathcal{X}$ we can write $x=s_{1}x_{1}+...+s_{k}x_{k}$
for some $k<\infty$, $x_{j}\in\mathcal{N}$ and $s_{j}>0$. Let $t=s_{1}%
+...+s_{k}$ and $p_{j}=s_{j}/t$ to have $x=ty$ with $y=p_{1}x_{1}%
+...+p_{k}x_{k}\in\mathcal{C}$.
The case $x=0$ is trivial.
\end{proof}

\section{Pointed spaces and absorption\label{AfdGepunt}}

The absorbing property in Definition \ref{eenhBal} is needed for the function
$\left\|  \cdot\right\|  _{\mathcal{C}}$ on $\mathcal{X}$ in Proposition
\ref{normUitBal} to be well defined and finite at every point of $\mathcal{X}%
$, and is in that sense the most basic property. In this section we show how
it arises under general assumptions for the set to be used as a unit ball for
the $W_{1}$ norm. We are ultimately interested in the situation where
$\mathcal{X}$ is a space of maps containing the channels between two composite
systems. In this section $\mathcal{X}$ is more generally taken as a vector
space $\mathcal{O}$ of linear maps between two vector spaces, each with a
distinguished point. The maps will be required to map the one distinguished
point to the other, as an abstraction of unitality. Actual unitality, in the
case of unital algebras as the pointed spaces, will be treated in the next
section, while the C*-algebraic case, for channels, follows in Section
\ref{AfdC*}.

\subsection{Pointed spaces}

By a \emph{pointed space} $(A,u_{A})$ we simply mean a real or complex vector
space $A$ with a distinguished non-zero element $u_{A}\in A$. For simplicity
of notation, a pointed space $(A,u_{A})$ will be denoted as $A$, with $u_{A}$
assumed as the notation for the distinguished point.

Keep in mind that eventually (Section \ref{AfdC*}) the pointed spaces will be
taken to be (complex) unital C*-algebras, with the units serving as the
distinguished points. These C*-algebras will generalize the matrix algebras
from Section \ref{AfdBuitL}.

In the remainder, all pointed spaces involved are assumed to be over the same
scalars, either real or complex. However, in either case certain constructions
will involve the span of a subset over real scalars, leading to a real vector
space, and such spans will be indicated by $\operatorname*{span}%
\nolimits_{\mathbb{R}}$.

For any two vector spaces $A$ and $B$ over the same scalars, the space of all
linear maps $\lambda:A\rightarrow B$ will be denoted by $L(A,B)$. When $A$ and
$B$ are pointed spaces, we define
\[
L_{u}(A,B)=\{\eta\in L(A,B):\eta(u_{A})=u_{B}\},
\]
which is the set of \emph{pointed maps} from $A$ to $B$. This is of course an
abstract version of unital maps in the case where $A$ and $B$ are unital algebras.

The following result will shortly be used in tandem with Lemma
\ref{abstrakteAbs} to prove the absorbing property of certain sets.

\begin{lemma}
\label{Ospan}Let $A$ and $B$ be pointed spaces and consider any subset
$\mathcal{L}$ of $L_{u}(A,B)$. Set%
\[
\mathcal{O}:=\{\lambda\in\operatorname*{span}\nolimits_{\mathbb{R}}%
\mathcal{L}:\lambda(u_{A})=0\}\text{ \ and \ }\mathcal{V}:=\{\eta-\theta
:\eta,\theta\in\mathcal{L}\}.
\]
Then it follows that $\mathcal{O}=\operatorname*{span}\nolimits_{\mathbb{R}%
}\mathcal{V}$.
\end{lemma}

\begin{proof}
Clearly $\operatorname*{span}\nolimits_{\mathbb{R}}\mathcal{V}\subset
\mathcal{O}$. Conversely, take any non-zero $\lambda\in\mathcal{O}$ and write
$\lambda_{1}:=\lambda$. Since $\lambda_{1}(u_{A})=0$ and $u_{B}\neq0$, the
coefficients in $\lambda_{1}$'s expansion as a linear combination of
$\mathcal{L}$'s elements, can not all be strictly positive or all strictly
negative. Without loss of generality we can therefore rewrite the expansion in
the form%
\[
\lambda=r_{1}^{(1)}\eta_{1}^{(1)}+...+r_{k^{(1)}}^{(1)}\eta_{k^{(1)}}%
^{(1)}-s_{1}^{(1)}\theta_{1}^{(1)}-...-s_{k^{(1)}}^{(1)}\theta_{k^{(1)}}^{(1)}%
\]
where $\eta_{j}^{(1)},\theta_{j}^{\left(  1\right)  }\in\mathcal{L}$ and
$r_{j}^{(1)},s_{j}^{(1)}>0$ for some (finite) $k^{(1)}$, where the superscript
$(1)$ is merely an index, and with the proviso that $\lambda(u_{A})=0$. Note
that $k^{(1)}$ is simply the biggest of the number of positive coefficients
and the number of negative coefficients in $\lambda_{1}$'s initial expansion,
while the coefficients in smaller number are split to increase their number to
$k^{(1)}$. Writing $\lambda_{1}^{\prime}=r_{1}^{(1)}(\eta_{1}^{(1)}-\theta
_{1}^{(1)})+...+r_{k^{(1)}}^{(1)}(\eta_{k^{(1)}}^{(1)}-\theta_{k^{(1)}}%
^{(1)})\in\operatorname*{span}\nolimits_{\mathbb{R}}\mathcal{V}$ and
$\lambda_{2}=(r_{1}^{(1)}-s_{1}^{(1)})\theta_{1}^{(1)}+...+(r_{k^{(1)}}%
^{(1)}-s_{k^{(1)}}^{(1)})\theta_{k^{(1)}}^{(1)}$, we have $\lambda_{1}%
=\lambda_{1}^{\prime}+\lambda_{2}$, implying that $\lambda_{2}\in\mathcal{O}$.
If $\lambda_{2}\neq0$, then repeat this procedure for $\lambda_{2}$ instead of
$\lambda_{1}$, noting that we can now analogously write%
\[
\lambda_{2}=r_{1}^{(2)}\eta_{1}^{(2)}+...+r_{k^{(2)}}^{(2)}\eta_{k^{(2)}%
}^{(2)}-s_{1}^{(2)}\theta_{1}^{(2)}-...-s_{k^{(2)}}^{(2)}\theta_{k^{(2)}%
}^{(2)},
\]
but with $k^{(2)}<k^{(1)}$. This delivers $\lambda_{2}^{\prime}$
$\in\operatorname*{span}\nolimits_{\mathbb{R}}\mathcal{V}$ and $\lambda_{3}%
\in\mathcal{O}$. If $\lambda_{3}\neq0$, then repeat for $\lambda_{3}$, etc.
Since $k^{(j+1)}<k^{(j)}$, this process must stop to deliver $\lambda
=\lambda_{1}^{\prime}+...+\lambda_{m}^{\prime}\in\operatorname*{span}%
\nolimits_{\mathbb{R}}\mathcal{V}$.
\end{proof}

Note that the lemma is geared towards spaces such as that appearing in
(\ref{spesO}), but in the Heisenberg picture.

\subsection{The composite setup\label{OndAfdOpset}}

Here we set up a framework which will serve as an abstraction of composite
systems in terms of pointed spaces. We also define the set for which the
absorbing property will be proven in the next subsection, leading to it being
a semi unit ball giving a $W_{1}$ seminorm.

Consider pointed spaces $A_{1},...,A_{n}$ and $B_{1},...,B_{n}$. Let
\[
A=A_{1}\odot...\odot A_{n}%
\]
be the algebraic tensor product of the vector spaces $A_{1},...,A_{n}$, which
is itself a pointed space with distinguished point%
\[
u_{A}:=u_{A_{1}}\otimes...\otimes u_{A_{n}}.
\]
However, as we want to allow for completions of $A$, in particular in the
C*-algebraic framework of Section \ref{AfdC*} (where $A_{1},...,A_{n}$ will be
unital C*-algebras), we need to allow for completions of the tensor product of
$B_{1},...,B_{n}$ from the outset. To emphasize this, we write%
\[
B=B_{1}\otimes...\otimes B_{n},
\]
which may be the algebraic tensor product, or some completion of it with
respect to a specified norm. In all cases $B$ is a pointed space with
$u_{B}=u_{B_{1}}\otimes...\otimes u_{B_{n}}$. The particular tensor product
$B$ remains fixed throughout this section, however. The point of this is that
in the theory developed here, linear maps $A\rightarrow B$ are then not
constrained to have their image contained in an uncompleted tensor product.

Along the lines of Section \ref{AfdBuitL}, we define%
\[
A_{\widehat{j}}=A_{1}\odot...\widehat{A}_{j}...\odot A_{n}\text{ \ and
\ }B_{\widehat{j}}=B_{1}\otimes...\widehat{B}_{j}...\otimes B_{n},
\]%
\[
A_{\leq j}=A_{1}\odot...\odot A_{j}\text{ \ and \ }A_{\geq j}=A_{j}%
\odot...\odot A_{n},
\]
and%
\[
B_{\leq j}=B_{1}\otimes...\otimes B_{j}\text{ \ and \ }B_{\geq j}=B_{j}%
\otimes...\otimes B_{n},
\]
for $j=1,...,n$, using the same completion (if relevant) for tensor products
of the $B_{j}$'s as for $B$.

Let $\nu_{j}$ be a linear functional on $B_{j}$ such that
\[
\nu_{j}(u_{B_{j}})=1
\]
for $j=1,...,n$. These functionals will act as reference functionals relative
to which linear maps $A\rightarrow B$ will be reduced to $A_{\widehat{j}%
}\rightarrow B_{\widehat{j}}$. Later on (Section \ref{AfdC*}), in the case of
unital C*-algebras, the $\nu_{j}$'s will be taken as states, generalizing the
normalized traces used in Section \ref{AfdBuitL}.

If $B$ is indeed a completion with respect to some norm, rather than just an
algebraic tensor product, we also assume that algebraic tensor products of the
$\nu_{j}$'s with themselves as well as with the identity maps
$\operatorname*{id}\nolimits_{B_{\leq j-1}}$ and $\operatorname*{id}%
\nolimits_{B_{\geq j+1}}$ are continuous with respect to this norm and
therefore uniquely extendible to the completed tensor products. The relevant
reductions of a linear map $\eta\in L(A,B)$ from $A$ to $B$ can then be
defined as%
\begin{align}
\eta_{\widehat{j}}  &  =(\operatorname*{id}\nolimits_{B_{\leq j-1}}\otimes
\nu_{j}\otimes\operatorname*{id}\nolimits_{B_{\geq j+1}})\circ\eta
|_{A_{\widehat{j}}}:A_{\widehat{j}}\rightarrow B_{\widehat{j}}%
\label{jKomplRed}\\
\eta_{\leq j}  &  =(\operatorname*{id}\nolimits_{B_{\leq j}}\otimes\nu
_{>j})\circ\eta|_{A_{\leq j}}:A_{\leq j}\rightarrow B_{\leq j}\label{<=jRed}\\
\eta_{\geq j}  &  =(\nu_{<j}\otimes\operatorname*{id}\nolimits_{B_{\geq j}%
})\circ\eta|_{A_{\geq j}}:A_{\geq j}\rightarrow B_{\geq j} \label{>=jRed}%
\end{align}
for $j=1,...,n$, with $\nu_{<j}=\nu_{1}\otimes...\otimes\nu_{j-1}$ (empty when
$j=1$) and $\nu_{>j}=\nu_{j+1}\otimes...\otimes\nu_{n}$ (empty when $j=n$),
where all the indicated tensor product maps are algebraic when the tensor
product $B$ is, or (uniquely) continuously extended to $B$ as assumed possible
above, when $B$ is completed. Here the restriction $\eta|_{A_{\widehat{j}}%
}:A_{\widehat{j}}\rightarrow B$ is defined via
\[
\eta|_{A_{\widehat{j}}}(a_{1}\otimes...\widehat{a}_{j}...\otimes a_{n}%
)=\eta(a_{1}\otimes...\otimes a_{j-1}\otimes u_{A_{j}}\otimes a_{j+1}%
\otimes...\otimes a_{n})
\]
for $a_{1}\in A_{1},...,a_{n}\in A_{n}$ (using the universal property),
replacing $a_{j}$ by $u_{A_{j}}$, and analogously for $\eta|_{A_{\leq j}%
}:A_{\leq j}\rightarrow B$ and $\eta|_{A_{\geq j}}:A_{\geq j}\rightarrow B$,
where the notation $\widehat{a}_{j}$ indicates the absence of $a_{j}$ in the
elementary tensor.

Note in particular that when $\eta\in L_{u}(A,B)$, it follows directly from
these definitions that%
\[
\eta_{\widehat{j}}\in L_{u}(A_{\widehat{j}},B_{\widehat{j}})\text{, \ }%
\eta_{\leq j}\in L_{u}(A_{\leq j},B_{\leq j})\text{ \ and \ }\eta_{\geq j}\in
L_{u}(A_{\geq j},B_{\geq j})\text{,}%
\]
where the property $\nu_{j}(u_{B_{j}})=1$ has been used. Note that for
$\eta,\theta\in L_{u}(A,B)$ this implies that%
\begin{equation}
\theta_{\leq j}\odot\eta_{\geq j+1}\in L_{u}(A,B) \label{absTensUn}%
\end{equation}
for $j=1,...,n-1$, which will implicitly play a role in Lemma \ref{teleskoop} below.

We also fix any subset
\[
\mathcal{L}\subset L_{u}(A,B)
\]
and let%
\begin{equation}
\mathcal{O}:=\{\lambda\in\operatorname*{span}\nolimits_{\mathbb{R}}%
\mathcal{L}:\lambda(u_{A})=0\} \label{O}%
\end{equation}
as in Lemma \ref{Ospan}. Our ultimate goal, given sufficient additional
structure and assumptions, is to define a metric, namely the Wasserstein
distance of order 1, on $\mathcal{L}$. This will be done by first obtaining a
seminorm, and in the next section under further assumptions a norm, on
$\mathcal{O}$.

Let%
\begin{equation}
\mathcal{N}_{j}:=\{\eta-\theta:\eta,\theta\in\mathcal{L}\text{ such that }%
\eta_{\widehat{j}}=\theta_{\widehat{j}}\}, \label{N_j}%
\end{equation}
generalizing (\ref{spesN_j}), though now in the Heisenberg picture, and set%
\begin{equation}
\mathcal{N}:=\bigcup_{j=1}^{n}\mathcal{N}_{j}\text{ \ and \ }\mathcal{C}%
:=\operatorname*{conv}\mathcal{N}, \label{bal}%
\end{equation}
where $\mathcal{C}$ is an abstract version of the set we ultimately want to
use as a unit ball defining the $W_{1}$ norm. In this section we settle for a seminorm.


\subsection{Absorption and $W_{1}$ seminorms}

In terms of the setup of the previous subsection, we now show that
$\mathcal{C}$ is absorbing for $\mathcal{O}$. The abstract assumption
(\ref{tensProdBehoud}) below, is made in lieu of complete positivity. Keep in
mind that because of (\ref{absTensUn}), $\theta_{\leq j}\odot\eta_{\geq j+1}$
in (\ref{tensProdBehoud}) is already a pointed map (which is an abstraction of
a unital map). The main technical step is the following lemma.

\begin{lemma}
\label{teleskoop}Assume that%
\begin{equation}
\theta_{\leq j}\odot\eta_{\geq j+1}\in\mathcal{L} \label{tensProdBehoud}%
\end{equation}
for all $\eta,\theta\in\mathcal{L}$ and $j=1,...,n-1$.\ For any $\eta
,\theta\in\mathcal{L}$ it then follows that $\eta-\theta=\lambda
_{1}+...+\lambda_{n}$ for some $\lambda_{j}\in\mathcal{N}_{j}$. In particular,%
\[
\mathcal{O}=\operatorname*{span}\nolimits_{\mathbb{R}}\mathcal{N}\text{.}%
\]
\end{lemma}

\begin{proof}
In line with Lemma \ref{Ospan}, we set \ $\mathcal{V}:=\{\eta-\theta
:\eta,\theta\in\mathcal{L}\}$. The case $n=1$ indeed follows immediately from
Lemma \ref{Ospan}, as we then have $\mathcal{N}=\mathcal{V}$. We can therefore
assume $n>1$. For any $\eta,\theta\in\mathcal{L}$, set%
\begin{align*}
\lambda_{1}  &  =\eta-\theta_{\leq1}\odot\eta_{\geq2}\\
\lambda_{2}  &  =\theta_{\leq1}\odot\eta_{\geq2}-\theta_{\leq2}\odot\eta
_{\geq3}\\
&  \vdots\\
\lambda_{n-1}  &  =\theta_{\leq n-2}\odot\eta_{\geq n-1}-\theta_{\leq
n-1}\odot\eta_{\geq n}\\
\lambda_{n}  &  =\theta_{\leq n-1}\odot\eta_{\geq n}-\theta.
\end{align*}
Then $\eta-\theta=\lambda_{1}+...+\lambda_{n}$. If indeed $\lambda_{j}%
\in\mathcal{N}_{j}$, it follows that $\eta-\theta\in\operatorname*{span}%
\nolimits_{\mathbb{R}}\mathcal{N}$, hence $\mathcal{V\subset}%
\operatorname*{span}\nolimits_{\mathbb{R}}\mathcal{N}$. From $\mathcal{N}%
\subset\mathcal{V}$ and Lemma \ref{Ospan} we can conclude that $\mathcal{O}%
=\operatorname*{span}\nolimits_{\mathbb{R}}\mathcal{N}$.

It remains to show that $\lambda_{j}\in\mathcal{N}_{j}$. Because of
(\ref{tensProdBehoud}), we simply have to check the equality of the reductions
as required in (\ref{N_j}), i.e., that $(\lambda_{j})_{\widehat{j}}=0$ for
$j=1,...,n$. To handle all cases at once, set $A_{\leq0}=\mathbb{C}$, $A_{\geq
n+1}=\mathbb{C}$, $\theta_{\leq0}=\operatorname*{id}_{\mathbb{C}}$ and
$\eta_{\geq n+1}=\operatorname*{id}_{\mathbb{C}}$, making $\eta=\theta_{\leq
0}\odot\eta_{\geq1}$ and $\theta=\theta_{\leq n}\odot\eta_{\geq n+1}$. For
$j=1,...,n$, and any $a_{<j}\in A_{\leq j-1}$ and $a_{>j}\in A_{\geq j+1}$,
one has the following direct calculation:%
\begin{align*}
&  \left(  \theta_{\leq j-1}\odot\eta_{\geq j}\right)  _{\widehat{j}}%
(a_{<j}\otimes a_{>j})\\
&  =(\operatorname*{id}\nolimits_{B_{\leq j-1}}\otimes\nu_{j}\otimes
\operatorname*{id}\nolimits_{B_{\geq j+1}})\circ\left(  \theta_{\leq j-1}%
\odot\eta_{\geq j}\right)  (a_{<j}\otimes u_{A_{j}}\otimes a_{>j})\\
&  =\theta_{\leq j-1}(a_{<j})\otimes\left(  (\nu_{j}\otimes\operatorname*{id}%
\nolimits_{B_{\geq j+1}})\circ\eta_{\geq j}(u_{A_{j}}\otimes a_{>j})\right)
\end{align*}
and
\begin{align*}
&  (\nu_{j}\otimes\operatorname*{id}\nolimits_{B_{\geq j+1}})\circ\eta_{\geq
j}(u_{A_{j}}\otimes(\cdot))\\
&  =(\nu_{j}\otimes\operatorname*{id}\nolimits_{B_{\geq j+1}})\circ(\nu
_{<j}\otimes\operatorname*{id}\nolimits_{B_{\geq j}})\circ\eta|_{A_{\geq j}%
}(u_{A_{j}}\otimes(\cdot))\\
&  =(\nu_{<j+1}\otimes\operatorname*{id}\nolimits_{B_{\geq j+1}})\circ
\eta|_{A_{\geq j+1}}\\
&  =\eta_{\geq j+1},
\end{align*}
hence $\left(  \theta_{\leq j-1}\odot\eta_{\geq j}\right)  _{\widehat{j}%
}=\theta_{\leq j-1}\odot\eta_{\geq j+1}$.

Similarly $\left(  \theta_{\leq j}\odot\eta_{\geq j+1}\right)  _{\widehat{j}%
}=\theta_{\leq j-1}\odot\eta_{\geq j+1}$, thus $\left(  \theta_{\leq j-1}%
\odot\eta_{\geq j}\right)  _{\widehat{j}}=\left(  \theta_{\leq j}\odot
\eta_{\geq j+1}\right)  _{\widehat{j}}$, as required.
\end{proof}

Using this lemma we can now show that $\mathcal{C}$ is indeed absorbing.

\begin{proposition}
\label{absorberend}Assuming (\ref{tensProdBehoud}), the set $\mathcal{C}$ in
(\ref{bal}) is a semi unit ball (Definition \ref{eenhBal}) for $\mathcal{O}$
given by (\ref{O}).
\end{proposition}

\begin{proof}
From (\ref{N_j}) and (\ref{bal}) it is clear that $-\mathcal{N}=\mathcal{N}$,
i.e., $\mathcal{N}$ is symmetric, hence so is $\mathcal{C}$, which is also
convex by definition. In addition, because of $\mathcal{N}$'s symmetry, Lemmas
\ref{teleskoop} and \ref{abstrakteAbs} imply that $\mathcal{C}$ is absorbing
for $\mathcal{O}$.
\end{proof}

Because of Proposition \ref{normUitBal}, this is sufficient to deliver a
seminorm. We summarize this as follows.

\begin{definition}
\label{gepunteStruk}The structure set up in Subsection \ref{OndAfdOpset}, with
$\mathcal{L}$ assumed to satisfy (\ref{tensProdBehoud}), is called a
\emph{pointed }$(W_{1},n)$\emph{ structure} or \emph{pointed }$W_{1}$\emph{
structure} (if $n$ is clear from context), and is denoted by the shorthand
$(A,B,\nu,\mathcal{L})$, where $\nu:=(\nu_{1},...\nu_{n})$. The rest of the
notation in Subsection \ref{OndAfdOpset} is then implied. For clarity the
space $\mathcal{O}$ and set $\mathcal{C}$ in (\ref{O}) and (\ref{bal}) can
respectively be denoted by
\[
\mathcal{O}_{\mathcal{L}}\text{ \ and \ }\mathcal{C}_{\mathcal{L}}%
\]
in this context.
\end{definition}

\begin{corollary}
\label{W-seminorm}Given a pointed $(W_{1},n)$ structure $(A,B,\nu
,\mathcal{L})$, the function $\left\|  \cdot\right\|  _{W_{1}}$ on
$\mathcal{O}_{\mathcal{L}}$ defined by
\[
\left\|  \lambda\right\|  _{W_{1}}=\inf\{t\geq0:\lambda\in t\mathcal{C}%
_{\mathcal{L}}\},
\]
for every $\lambda\in\mathcal{O}_{\mathcal{L}}$, is a seminorm referred to as
the $W_{1}$\emph{ seminorm} associated to $(A,B,\nu,\mathcal{L})$.
Consequently, the function $W_{1}:\mathcal{L}\times\mathcal{L}\rightarrow
\mathbb{R}$ defined by%
\[
W_{1}(\eta,\theta)=\left\|  \eta-\theta\right\|  _{W_{1}},
\]
is a pseudometric on $\mathcal{L}$.
\end{corollary}

Although $W_{1}$ is at this stage only a pseudometric, which means we may have
$W_{1}(\eta,\theta)=0$ for $\eta\neq\theta$, it will nevertheless be called
the \emph{Wasserstein distance of order 1} associated to $(A,B,\nu
,\mathcal{L})$. We still need ray-wise boundedness to achieve a norm and make
$W_{1}$ a metric, which is what we turn to next.

\section{Algebras, unital maps and ray-wise boundedness\label{AfdAlg}}

To obtain a norm from the $W_{1}$ seminorm in Corollary \ref{W-seminorm}, we
need ray-wise boundedness. In order to achieve this, we make use of a more
specialized algebraic framework as well as assumptions complementary to those
made in the previous section. We are going to work in the context of unital
algebras. This means that we do not assume the presence of an involution (an
adjoint operation) in the algebras. Consequently, as in the previous section,
positivity plays no role here, though we again make the abstract assumption
(\ref{tensProdBehoud}) which will be implied by the complete positivity of
channels in the next section. On the other hand, the unitality of maps will be used.

\subsection{Ray-wise boundedness in an algebraic framework}

For the moment we step away from the tensor product setup of the previous
section, and consider a simple algebraic setting. The core structure (\ref{O})
will remain in place, however. This allows us to obtain the remaining building
block required by Proposition \ref{normUitBal} in the next result. We return
to the tensor product setup in the next subsection.

Note that in the proposition below, all elements of $\mathcal{L}$ are unital maps.

\begin{proposition}
\label{straalBeg}Let $A$ and $B$ be any unital algebras (both of them real, or
both of them complex), with their units $1_{A}$ and $1_{B}$ respectively
serving as the distinguished points making $A$ and $B$ pointed spaces. Fix any
subset $\mathcal{L}$ of $L_{u}(A,B)$ and set%
\[
\mathcal{O}:=\{\lambda\in\operatorname*{span}\nolimits_{\mathbb{R}}%
\mathcal{L}:\lambda(1_{A})=0\}\text{ \ and \ }\mathcal{V}:=\{\eta-\theta
:\eta,\theta\in\mathcal{L}\}.
\]
Consider any subset $\mathcal{N}$ of $\mathcal{V}$ and set $\mathcal{C}%
:=\operatorname*{conv}\mathcal{N}$. Then $\mathcal{C}$ is ray-wise bounded in
$\mathcal{O}$.
\end{proposition}

\begin{proof}
In order to decide whether an element of $\mathcal{O}$ lies outside
$\mathcal{C}$, it is going to be convenient to attach a quantitative measure
to the element, which when too large, will imply that the element is not in
$\mathcal{O}$. To attain this, we follow a construction inspired by the
representation theory of C*-algebras.

For every $\eta\in\mathcal{L}$, define a bilinear map $\left\langle
\cdot,\cdot\right\rangle _{\eta}:A\times A\rightarrow B$ by%
\[
\left\langle x,y\right\rangle _{\eta}=\eta(xy)
\]
for all $x,y\in A$. We write $G_{\eta}$ for $A$ equipped with this bilinear
map. Define $\pi_{\eta}:A\rightarrow L(G_{\eta})$, with $L(G_{\eta})$ the
space of linear maps from $G_{\eta}$ to itself, through%
\[
\pi_{\eta}(a)x=ax
\]
for all $a\in A$ and $x\in G_{\eta}$, with $ax$ simply being the product in
$A$. Then
\[
\eta(a)=\left\langle 1_{A},\pi_{\eta}(a)1_{A}\right\rangle _{\eta}%
\]
in analogy to a cyclic representation obtained from the GNS construction,
where $1_{A}$ serves as the ``cyclic vector'' (indeed, $\pi_{\eta}%
(A)1_{A}=G_{\eta}$). We proceed to consider the direct sum%
\[
(G,\left\langle \cdot,\cdot\right\rangle ,\pi)=\bigoplus\limits_{\eta
\in\mathcal{L}}(G_{\eta},\left\langle \cdot,\cdot\right\rangle _{\eta}%
,\pi_{\eta}).
\]
I.e., every $x\in G$ is of the form $x=(x_{\eta})_{\eta\in\mathcal{L}}$ with
$x_{\eta}\in G_{\eta}$ and $\{\eta\in\mathcal{L}:x_{\eta}\neq0\}$ a finite
set. Furthermore, $\pi:A\rightarrow L(G)$ is defined by $\pi(a)x=(\pi_{\eta
}(a)x_{\eta})_{\eta\in\mathcal{L}}$ for all such $x$ and every $a\in A$.
Lastly, $\left\langle x,y\right\rangle :=\sum_{_{\eta\in\mathcal{L}}%
}\left\langle x_{\eta},y_{\eta}\right\rangle _{\eta}$ for all $x=(x_{\eta
})_{\eta\in\mathcal{L}},y=(y_{\eta})_{\eta\in\mathcal{L}}\in G$.

For any $\eta\in\mathcal{L}$, define $\hat{\eta}\in G$ by
\[
\hat{\eta}_{\theta}=\left\{
\begin{array}
[c]{ll}%
1_{A} & \text{for }\theta=\eta\\
0 & \text{for }\theta\neq\eta
\end{array}
\right.
\]
for all $\theta\in\mathcal{L}$. Note that $\left\langle \hat{\eta},\pi
(a)\hat{\eta}\right\rangle =\eta(a)$ for all $\eta\in\mathcal{L}$ and $a\in A$.

Consider any $\lambda\in\mathcal{O}$, which using Lemma \ref{Ospan}, we can
write as%
\[
\lambda=\sum_{j=1}^{l}r_{j}(\eta_{j}-\theta_{j})
\]
where $r_{j}\geq0$ and $\eta_{j},\theta_{j}\in\mathcal{L}$. Then%
\[
\lambda=\sum_{j=1}^{l}r_{j}\left(  \left\langle \hat{\eta}_{j},\pi(\cdot
)\hat{\eta}_{j}\right\rangle -\left\langle \hat{\theta}_{j},\pi(\cdot
)\hat{\theta}_{j}\right\rangle \right)  .
\]
We can use this to lift $\lambda:A\rightarrow B$ to a linear map $\bar
{\lambda}:L(G)\rightarrow B$ defined by%
\[
\bar{\lambda}(T)=\sum_{j=1}^{l}r_{j}\left(  \left\langle \hat{\eta}_{j}%
,T\hat{\eta}_{j}\right\rangle -\left\langle \hat{\theta}_{j},T\hat{\theta}%
_{j}\right\rangle \right)
\]
for all $T\in L(G)$. This lifting may not be unique (it may depend on the
choice of $\eta_{j}$'s and $\theta_{j}$'s), but for our purposes any such
lifting will do. In particular, any $\gamma\in\mathcal{C}$ can be lifted to
$\bar{\gamma}:L(G)\rightarrow B$ in the form%
\[
\bar{\gamma}(T)=\sum_{j=1}^{m}p_{j}\left(  \left\langle \hat{\alpha}_{j}%
,T\hat{\alpha}_{j}\right\rangle -\left\langle \hat{\beta}_{j},T\hat{\beta}%
_{j}\right\rangle \right)
\]
for some $\alpha_{j},\beta_{j}\in\mathcal{L}$, with $p_{j}\geq0$ and
$p_{1}+...+p_{m}=1$.

According to the Hahn-Banach theorem there is a linear functional $f$ on $B$
such that $f(1_{B})\neq0$ (in the case where $1_{B}=0$ and therefore
$B=\{0\}$, Proposition \ref{straalBeg} is trivial). Normalize it to obtain
\[
\nu=\frac{1}{f(1_{B})}f.
\]
For any $x,y\in G_{\eta}$ we use this to define $x\Join y:G_{\eta}\rightarrow
G_{\eta}$ by $(x\Join y)z=x\nu(\left\langle y,z\right\rangle _{\eta})$ for all
$z\in G_{\eta}$, where the notation $x\Join y$ is inspired by the Dirac
notation $\ket{x}  \bra{y}  $. For $x,y\in G$ this in turn allows us to define%
\[
x\Join_{\oplus}y=\bigoplus_{\zeta\in\mathcal{L}}x_{\eta}\Join y_{\eta}\in
L(G),
\]
i.e., $(x\Join_{\oplus}y)z=(x_{\eta}\nu(\left\langle y_{\eta},z_{\eta
}\right\rangle _{\eta}))_{\eta\in\mathcal{L}}$ for $z\in G$. For all
$\eta,\theta\in\mathcal{L}$ we then have
\begin{align*}
\left\langle \hat{\theta},(\hat{\eta}\Join_{\oplus}\hat{\eta})\hat{\theta
}\right\rangle  &  =\sum_{_{\zeta\in\mathcal{L}}}\left\langle \hat{\theta
}_{\zeta},\hat{\eta}_{\zeta}\nu\left(  \left\langle \hat{\eta}_{\zeta}%
,\hat{\theta}_{\zeta}\right\rangle _{\zeta}\right)  \right\rangle _{\zeta}\\
&  =\left\{
\begin{array}
[c]{ll}%
1_{B} & \text{for }\theta=\eta\\
0 & \text{for }\theta\neq\eta.
\end{array}
\right.
\end{align*}
Thus, in terms of $\bar{\gamma}$ above,
\[
\bar{\gamma}(\hat{\eta}\Join_{\oplus}\hat{\eta})=\sum_{j=1}^{m}p_{j}b_{j}%
\]
where $b_{j}\in\{-1_{B},0,1_{B}\}$ for all $j$, and consequently%
\begin{equation}
\left|  \nu\left(  \bar{\gamma}(\hat{\eta}\Join_{\oplus}\hat{\eta})\right)
\right|  \leq1 \label{kwantBalVwd}%
\end{equation}
for all $\eta\in\mathcal{L}$. This is therefore a condition satisfied by all
$\gamma\in\mathcal{C}$, for any lifting $\bar{\gamma}$ of the form above.

On the other hand, for any non-zero $\lambda\in\mathcal{O}$ and a lifting
$\bar{\lambda}$ as above, we see that $\overline{s\lambda}:=s\bar{\lambda}$
correspondingly lifts $s\lambda$ for any $s>0$, using $sr_{j}$ instead of
$r_{j}$. To simplify the notation in the remainder of the proof, rewrite
$\bar{\lambda}$ above as%
\[
\bar{\lambda}(T)=\sum_{i=1}^{k}q_{i}\left\langle \hat{\zeta}_{i},T\hat{\zeta
}_{i}\right\rangle ,
\]
where $k=2l$, $\zeta_{j}=\eta_{j}$, $\zeta_{j+k}=\theta_{j}$, $q_{j}=r_{j}$
and $q_{j+k}=-r_{j}$ for $j=1,...,k$. In this form we may as well assume
without loss that $\zeta_{i}\neq\zeta_{i^{\prime}}$ when $i\neq i^{\prime}$ by
collecting terms with $\zeta_{i}=\zeta_{i^{\prime}}$ if needed. As
$\lambda\neq0$, we have $q_{i}\neq0$ for some $i$. In terms of this we have%
\[
\overline{s\lambda}(\hat{\zeta}_{i}\Join_{\oplus}\hat{\zeta}_{i})=sq_{i}%
1_{B},
\]
hence%
\[
\left|  \nu\left(  \overline{s\lambda}(\hat{\zeta}_{i}\Join_{\oplus}\hat
{\zeta}_{i})\right)  \right|  =s|q_{i}|,
\]
for all $s>0$. It follows that there is an $s_{0}>0$ such that
\[
\left|  \nu\left(  \overline{s\lambda}(\hat{\zeta}_{i}\Join_{\oplus}\hat
{\zeta}_{i})\right)  \right|  >1
\]
i.e., $s\lambda\notin\mathcal{C}$ because of (\ref{kwantBalVwd}), for all
$s>s_{0}$, proving that $\mathcal{C}$ is ray-wise bounded according to
Definition \ref{eenhBal}(3).
\end{proof}

With this proposition we have all the elements of the abstract theory in
place, which will now allow us to formulate an abstract version of Wasserstein
distance of order 1 as a metric, rather than just a pseudometric.

\subsection{The composite algebraic framework and $W_{1}$ norms}

We return to the setup in Subsection \ref{OndAfdOpset}, but specialize it as follows.

\begin{definition}
\label{algWn}An \emph{algebraic }$(W_{1},n)$\emph{ structure} $(A,B,\nu
,\mathcal{L})$ is a pointed $(W_{1},n)$ structure as in Definition
\ref{gepunteStruk}, where the pointed spaces $A_{1},...,A_{n}$ and
$B_{1},...,B_{n}$ are unital algebras (all of them real, or all of them
complex), with their units serving as their distinguished points,%
\[
u_{A_{j}}=1_{A_{j}}\text{ \ and \ }u_{B_{j}}=1_{B_{j}}%
\]
for $j=1,...,n$.
\end{definition}

From the preceding development we immediately conclude the following.

\begin{theorem}
\label{W-norm}Consider an algebraic $(W_{1},n)$ structure $(A,B,\nu
,\mathcal{L})$. Then $\left\|  \cdot\right\|  _{W_{1}}$ defined by
\[
\left\|  \lambda\right\|  _{W_{1}}=\inf\{t\geq0:\lambda\in t\mathcal{C}%
_{\mathcal{L}}\}
\]
for all $\lambda\in\mathcal{O}_{\mathcal{L}}$, is a norm on $\mathcal{O}%
_{\mathcal{L}}$, called the $W_{1}$\emph{ norm} associated to $(A,B,\nu
,\mathcal{L})$.
\end{theorem}

\begin{proof}
From Propositions \ref{absorberend} and \ref{straalBeg} we know that
$\mathcal{C}_{\mathcal{L}}$ is a unit ball for $\mathcal{O}_{\mathcal{L}}$, as
defined in Definition \ref{eenhBal}. By Proposition \ref{normUitBal} we are done.
\end{proof}

\begin{corollary}
In terms of Theorem \ref{W-norm}, the function $W_{1}:\mathcal{L}%
\times\mathcal{L}\rightarrow\mathbb{R}$ defined by%
\[
W_{1}(\eta,\theta)=\left\|  \eta-\theta\right\|  _{W_{1}}%
\]
is a metric on $\mathcal{L}$, called the \emph{Wasserstein distance of order
1} on $\mathcal{L}$ associated to the algebraic $(W_{1},n)$ structure
$(A,B,\nu,\mathcal{L})$.
\end{corollary}

This theorem and its corollary are the main results up to this point and
completes the development without the presence of any form of positivity
assumed of the $A\rightarrow B$ maps in $\mathcal{L}$. So far the maps in
$\mathcal{L}$ were only assumed to be linear and unital. In the next section
we add complete positivity in a C*-algebraic framework.

\section{C*-algebras and complete positivity\label{AfdC*}}

For unital C*-algebras $A$ and $B$, let%
\[
K(A,B)
\]
be the set of all channels $E:A\rightarrow B$ from $A$ to $B$, where a
\emph{channel} is a unital completely positive linear (u.c.p.) map. The goal
of this section is to define a Wasserstein distance of order 1 on $K(A,B)$.
Conventionally the term Wasserstein distance applies to states, including
(integrals with respect to) probability measures in the classical case, i.e.,
the case $B=\mathbb{C}$. However, here we\ use the same terminology for
channels as well, as already indicated in Section \ref{AfdBuitL}.

It is a fairly straightforward matter to apply Theorem \ref{W-norm} in a
C*-algebraic framework, essentially taking $\mathcal{L}$ in the previous two
sections to be $K(A,B)$, though there are some technical points regarding this
which will be made clear in the proof of Theorem \ref{C*-W-norm} below. In
order to make this section as directly accessible as possible, however, we
formulate the definitions and results without reference to Sections
\ref{AfdGepunt} and \ref{AfdAlg}. References to these two sections will only
appear in the proof of Theorem \ref{C*-W-norm}. We start with some notation
and conventions.

The tensor products of C*-algebras are not merely algebraic, but are completed
in some norm. Specifically, the tensor products in this section are either all
minimal tensor products or all maximal tensor products. These tensor products
will simply be indicated by the symbol $\otimes$. Accordingly for tensor
products of maps on C*-algebras. Some standard background regarding complete
positivity and tensor products of C*-algebras can be reviewed in
\cite[Subsections II.6.9 and II.9.7]{Bl}. In particular we note that tensor
products of channels are again channels, for both the minimal and maximal
tensor products.

For easy reference, we highlight the main structure which will be used:

\begin{definition}
\label{n-C*-stelsel}Consider unital C*-algebras $A_{1},...,A_{n}$ and
$B_{1},...,B_{n}$,\ as well as a state $\nu_{j}$ on $B_{j}$ for $j=1,...,n$.
This will be referred to as an $n$\emph{-composite C*-system}, denoted
$(A_{j},B_{j},\nu_{j}:j=1,...,n)$, with the notation%
\[
A=A_{1}\otimes...\otimes A_{n}\text{ \ and \ }B=B_{1}\otimes...\otimes B_{n}%
\]
being implied.
\end{definition}

Given such an $n$-composite C*-system, the following notation, along the lines
of Section \ref{AfdBuitL}, will be used to set up the Wasserstein norm and
subsequent Wasserstein distance of order 1:%

\[
A_{\widehat{j}}:=A_{1}\otimes...\widehat{A}_{j}...\otimes A_{n}\text{ \ and
\ }B_{\widehat{j}}:=B_{1}\otimes...\widehat{B}_{j}...\otimes B_{n},
\]
and%
\[
B_{\leq j}:=B_{1}\otimes...\otimes B_{j}\text{ \ and \ }B_{\geq j}%
:=B_{j}\otimes...\otimes B_{n},
\]
for $j=1,...,n$. Keep in mind that as mentioned above, either all of these
tensor products are minimal, or all are maximal. A channel $E:A\rightarrow B$
can be reduced to a channel $E_{\widehat{j}}:A_{\widehat{j}}\rightarrow
B_{\widehat{j}}$ by%
\[
E_{\widehat{j}}:=(\operatorname*{id}\nolimits_{B_{\leq j-1}}\otimes\nu
_{j}\otimes\operatorname*{id}\nolimits_{B_{\geq j+1}})\circ E|_{A_{\widehat
{j}}}%
\]
for $j=1,...,n$. Here the restriction $E|_{A_{\widehat{j}}}:A_{\widehat{j}%
}\rightarrow B$ is defined via
\[
E|_{A_{\widehat{j}}}(a_{1}\otimes...\widehat{a}_{j}...\otimes a_{n}%
)=E(a_{1}\otimes...\otimes a_{j-1}\otimes1_{A_{j}}\otimes a_{j+1}%
\otimes...\otimes a_{n})
\]
for $a_{1}\in A_{1},...,a_{n}\in A_{n}$, where $1_{A_{j}}$ denotes the unit of
$A_{j}$.

Set
\[
\mathcal{O}_{A,B}:=\{\lambda\in\operatorname*{span}\nolimits_{\mathbb{R}%
}K(A,B):\lambda(1_{A})=0\},
\]%
\[
\mathcal{N}_{j}:=\{D-E:D,E\in K(A,B)\text{ such that }D_{\widehat{j}%
}=E_{\widehat{j}}\},
\]
and%
\[
\mathcal{N}:=\bigcup_{j=1}^{n}\mathcal{N}_{j}\text{ \ and \ }\mathcal{C}%
:=\operatorname*{conv}\mathcal{N}\text{.}%
\]
Here of course two channels $D,E\in K(A,B)$ are said to be \emph{neighbouring}
if $D_{\widehat{j}}=E_{\widehat{j}}$ for some $j\in\{1,...,n\}$.

This allows us to state the main result of this section, which is one of the
main results of the paper.

\begin{theorem}
\label{C*-W-norm}Let $(A_{j},B_{j},\nu_{j}:j=1,...,n)$ be an $n$-composite
C*-system. Then in both the minimal and maximal tensor product setup,
$\left\|  \cdot\right\|  _{W_{1}}$ defined by%
\[
\left\|  \lambda\right\|  _{W_{1}}=\inf\{t\geq0:\lambda\in t\mathcal{C}\}
\]
for all $\lambda\in\mathcal{O}_{A,B}$, is a norm on $\mathcal{O}_{A,B}$,
called the $W_{1}$\emph{ norm} associated to $(A_{j},B_{j},\nu_{j}%
:j=1,...,n)$.
\end{theorem}

\begin{proof}
We are going to obtain this from the algebraic (rather than C*-algebraic)
setup of the previous section. Therefore we have to convert between algebraic
and completed tensor products as needed. Write $A_{\odot}=A_{1}\odot...\odot
A_{n}$ and define
\[
K(A,B)|_{\odot}=\{E|_{A_{\odot}}:E\in K(A,B)\}.
\]
Since a channel $E\in K(A,B)$ is necessarily continuous, it is uniquely
determined by its restriction $E|_{A_{\odot}}$, the latter being the usual
restriction of the map $E$ to the subset $A_{\odot}$ of $A$. Hence
$K(A,B)|_{\odot}$ and $K(A,B)$ are in one-to-one correspondence.

Also define
\[
A_{\odot,\leq j}=A_{1}\odot...\odot A_{j}\text{ \ and \ }A_{\odot,\geq
j}=A_{j}\odot...\odot A_{n},
\]
and%
\[
A_{\leq j}=A_{1}\otimes...\otimes A_{j}\text{ \ and \ }A_{\geq j}=A_{j}%
\otimes...\otimes A_{n},
\]
for $j=1,...,n$, as well as%

\begin{align*}
E_{\leq j}  &  =(\operatorname*{id}\nolimits_{B_{\leq j}}\otimes\nu_{>j})\circ
E|_{A_{\leq j}}:A_{\leq j}\rightarrow B_{\leq j}\\
E_{\geq j}  &  =(\nu_{<j}\otimes\operatorname*{id}\nolimits_{B_{\geq j}})\circ
E|_{A_{\geq j}}:A_{\geq j}\rightarrow B_{\geq j}%
\end{align*}
for any $E\in K(A,B)$.

Note that $(A_{j},B_{j},\nu_{j}:j=1,...,n)$ gives an algebraic $(W_{1},n)$
structure $(A_{\odot},B,\nu,K(A,B)|_{\odot})$ as in Definition \ref{algWn}.
This follows from the automatic continuity of tensor products of\ states and
identity maps in the C*-algebraic framework, along with the fact that
condition (\ref{tensProdBehoud}) is satisfied. The latter, in terms of
(\ref{<=jRed}) and (\ref{>=jRed}), being%
\[
\left(  E|_{A_{\odot}}\right)  _{\leq j}\odot\left(  D|_{A_{\odot}}\right)
_{\geq j+1}=E_{\leq j}|_{A_{\odot,\leq j}}\odot D_{\geq j+1}|_{A_{\odot,\geq
j+1}}\in K(A,B)|_{\odot}%
\]
for all $D,E\in K(A,B)$, where $|_{A_{\odot}}$, $|_{A_{\odot,\leq j}}$ and
$|_{A_{\odot,\geq j+1}}$ are the usual restrictions to the indicated algebraic
tensor products. This fact in turn is true, since $E_{\leq j}\odot D_{\geq
j+1}$ uniquely extends to an element $E_{\leq j}\otimes D_{\geq j+1}$ of
$K(A,B)$, as $E_{\leq j}$ and $D_{\geq j+1}$ themselves are channels (being
the composition of u.c.p. maps), hence indeed
\[
E_{\leq j}|_{A_{\odot,\leq j}}\odot D_{\geq j+1}|_{A_{\odot,\geq j+1}%
}=(E_{\leq j}\otimes D_{\geq j+1})|_{A_{\odot}}\in K(A,B)|_{\odot}.
\]

By Theorem \ref{W-norm} and the one-to-one correspondence between
$K(A,B)|_{\odot}$ and $K(A,B)$ mentioned above, we are done.
\end{proof}

This leads to the following key conclusion.

\begin{corollary}
Given an $n$-composite C*-system $(A_{j},B_{j},\nu_{j}:j=1,...,n)$, then in
both the minimal and the maximal tensor product setup we obtain a metric
$W_{1}$ on $K(A,B)$ defined by%
\[
W_{1}(D,E)=\left\|  D-E\right\|  _{W_{1}}%
\]
for all $D,E\in K(A,B)$, called the \emph{Wasserstein distance of order 1
associated to} $(A_{j},B_{j},\nu_{j}:j=1,...,n)$.
\end{corollary}

Note that Section \ref{AfdBuitL} emerges as a special case of this section,
albeit directly in the Heisenberg picture, by simply setting
\[
A_{j}=M_{q_{j}}\text{ \ and \ }B_{j}=M_{r_{j}}%
\]
and letting $\nu_{j}$ be the normalized trace on $B_{j}$.

Another special case is $B_{1}=...=B_{n}=\mathbb{C}$, with $\nu_{1}%
,...,\nu_{n}$ becoming trivial and irrelevant, but with general unital
C*-algebras $A_{1},...,A_{n}$. In this case $K(A,B)$ is the set of all states
on $A=A_{1}\otimes...\otimes A_{n}$, hence $W_{1}$ is now the Wasserstein
distance of order 1 between states on $A$. For $A_{1}=...=A_{n}=M_{d}$ this
reduces to the (quantum) Wasserstein distance of order 1 studied in
\cite{DMTL}, as can be seen from Section \ref{AfdBuitL}, keeping in mind that
a state on $A$ is exactly a normalized positive linear functional $\mu$, which
in this finite dimensional case can be uniquely represented as $\mu
(a)=\operatorname*{Tr}(\rho a)$ for all $a\in A$ in terms of some density
matrix $\rho$.

We have focussed on the composite system aspect of the framework. From a
single system point of view, note that in the finite dimensional setup for
states, and setting $n=1$, for any states $\psi$ and $\omega$ on $A_{1}=M_{d}%
$, we have
\[
W_{1}(\psi,\omega)=\frac{1}{2}\operatorname*{Tr}|\rho_{\psi}-\rho_{\omega}|,
\]
with $\rho_{\psi}$ and $\rho_{\omega}$ being the density matrices representing
$\psi$ and $\omega$ respectively, according to \cite[Proposition 2]{DMTL}. In
the general C*-algebraic case for states (i.e., $B=\mathbb{C}$) with $n=1$, we
can therefore view $W_{1}$ as an abstract version of the trace distance
between states, despite the fact that no canonical trace is specified on
$A_{1}$ in this setup.
Refer to Section \ref{AfdEinde} for further remarks related to this.


This ends our development of the Wasserstein distance of order 1. Next we
study its behaviour in relation to subsystems.

\section{Subsystems and additivity\label{AfdSubStelsels}}

A core idea behind $W_{1}$ is that it is built to reflect the composite
structure of systems. It is therefore natural to study its basic properties in
relation to subsystems of the composite systems, i.e., smaller tensor
products. This is what is done in the current section, first in terms of the
$W_{1}$ seminorms obtained for pointed $(W_{1},n)$-structures in Section
\ref{AfdGepunt}, and subsequently for the C*-algebraic framework of the
previous section. We focus on the additivity of $W_{1}$ with respect to tensor
products (see Theorems \ref{gepunteAdEiensk} and \ref{C*AdEiensk}) and the
resulting stability of $W_{1}$ (see Corollary \ref{stab}). The additivity
results of this section generalize those of \cite[Section IV.C]{DMTL}, though
the techniques to achieve them are necessarily different, as \cite{DMTL} makes
use of trace norms, which are not available in our context. The reader who
wants to see the main results in the C*-algebraic context, can turn directly
to Subsection \ref{OndAfdC*ad}, but the proofs and some notation rely on
Subsections \ref{OndAfdGepuntSubstr}, \ref{OndAfdRedStruk} and \ref{OndAfdAd}.

\subsection{Pointed $W_{1}$ substructures and their $W_{1}$
seminorms\label{OndAfdGepuntSubstr}}


Let $(A,B,\nu,\mathcal{L})$ be a pointed $(W_{1},n)$ structure as defined in
Definition \ref{gepunteStruk}, again writing $\mathcal{C}_{\mathcal{L}}$ for
the semi unit ball in $\mathcal{O}_{\mathcal{L}}$, as given by Proposition
\ref{absorberend}. For simplicity of notation, particularly in the following
subsections, however, we continue to write $\mathcal{N}_{j}$ and $\mathcal{N}$
as in Section \ref{AfdGepunt}, rather than, say, $\mathcal{N}_{\mathcal{L},j}$
and $\mathcal{N}_{\mathcal{L}}$. To describe the related subsystems, we need
some further notation.

Write
\[
\lbrack n]:=\{1,...,n\},
\]
and $\mathcal{P}_{n}$ for the collection of non-empty proper subsets $J$ of
$[n]$; by ``proper'' we mean that $J\neq\lbrack n]$. The complement of
$J\in\mathcal{P}_{n}$ will be written as
\[
J^{\prime}:=[n]\backslash J=\{j\in\lbrack n]:j\notin J\}\in\mathcal{P}_{n}.
\]

We define
\[
A_{J}=\bigodot_{j\in J}A_{j}\text{ \ and \ }B_{J}=\bigotimes_{j\in J}B_{j}%
\]
for any $J\in\mathcal{P}_{n}$, with the same tensor product conventions as in
Subsection \ref{OndAfdOpset}. Here the order of the $A_{j}$'s in $A_{J}$ is
taken to be the same as in $A=A_{1}\odot...\odot A_{n}$, for example
$A_{\{2,5\}}=A_{2}\odot A_{5}$ rather than $A_{5}\odot A_{2}$. Similarly for
$B_{J}$ and correspondingly for $\nu_{J}:=(\nu_{j})_{j\in J}$. Elementary
tensors in $A_{J}$ can be denoted as%
\[
\otimes_{j\in J}a_{j}%
\]
for $a_{j}\in A_{j}$, and similarly for $B_{J}$, for any $J\in\mathcal{P}_{n}%
$. In particular $A_{J}$ is a pointed space with distinguished point%
\[
u_{A_{J}}:=\otimes_{j\in J}u_{A_{j}}.
\]
Similarly for $B_{J}$.

We need to define corresponding reductions of maps. Given $\eta\in L(A,B)$,
its \emph{reduction}
\[
\eta^{J}:=\left(  \bigotimes_{j=1}^{n}\varphi_{j}\right)  \circ\eta|_{A_{J}%
}\in L(A_{J},B_{J})
\]
to $J\in\mathcal{P}_{n}$ (or over $J^{\prime}$) is defined as an obvious
generalization of the reductions considered in Subsection \ref{OndAfdOpset},
where
\[
\varphi_{j}=\left\{
\begin{array}
[c]{ll}%
\operatorname*{id}\nolimits_{B_{j}} & \text{for }j\in J\text{ }\\
\nu_{j} & \text{for }j\in J^{\prime}%
\end{array}
\right.
\]
and with $\eta|_{A_{J}}:A_{J}\rightarrow B$ given via%
\[
\eta|_{A_{J}}(\otimes_{j\in J}a_{j})=\eta(a_{1}\otimes...\otimes a_{n})
\]
for $\otimes_{j\in J}a_{j}$ an elementary tensor in $A_{J}$, by setting
$a_{j}=u_{A_{j}}$ for $j\in J^{\prime}$. For any subset $\mathcal{S}$ of
$L(A,B)$, let%
\begin{equation}
\mathcal{S}^{J}:=\{\eta^{J}:\eta\in\mathcal{S}\}. \label{S^J}%
\end{equation}
Note that one can of course reduce any $\theta\in\mathcal{S}^{J}$ to
$\theta^{I}$ for any non-empty proper $I\subset J$, by the obvious adjustment
of the method above to this case.

Then we have the next basic fact in terms of Definition \ref{gepunteStruk}.

\begin{proposition}
Let $(A,B,\nu,\mathcal{L})$ be a pointed $(W_{1},n)$ structure and consider
any $J\in\mathcal{P}_{n}$. The pointed spaces $A_{j}$ and $B_{j}$ for $j\in
J$, along with $\nu_{J}=(\nu_{j})_{j\in J}$ and $\mathcal{L}^{J}$, then form a
pointed $(W_{1},|J|)$ structure $(A_{J},B_{J},\nu_{J},\mathcal{L}^{J})$,
called a \emph{pointed }$W_{1}$\emph{ substructure} of $(A,B,\nu,\mathcal{L})$.
\end{proposition}

\begin{proof}
Clearly $\eta^{J}(u_{A_{J}})=u_{B_{J}}$ for every $\eta\in\mathcal{L}$ by
$\eta^{J}$'s definition, hence $\mathcal{L}^{J}\subset L(A_{J},B_{J})$, the
tensor products of $\nu_{j}$'s and identity maps restrict those of
$(A,B,\nu,\mathcal{L})$ and are therefore still continuous, while the analogue
of (\ref{tensProdBehoud}) is easily seen to hold in this context by simply
reducing it to $L(A_{J},B_{J})$. The latter is confirmed by a direct
calculation similar to that in Lemma \ref{teleskoop}'s proof.
\end{proof}

By Corollary \ref{W-seminorm}, the pointed $(W_{1},|J|)$ structure
$(A_{J},B_{J},\nu_{J},\mathcal{L}^{J})$ provides us with a $W_{1}$ seminorm on%
\begin{equation}
\mathcal{O}_{\mathcal{L}^{J}}:=\{\lambda\in\operatorname*{span}%
\nolimits_{\mathbb{R}}\mathcal{L}^{J}:\lambda(u_{A_{J}})=0\} \label{O_LJ}%
\end{equation}
for every $J\in\mathcal{P}_{n}$, still denoted as
\[
\left\|  \cdot\right\|  _{W_{1}},
\]
as well as the resulting pseudometric $W_{1}$ on $\mathcal{L}^{J}$. This is
simply a case of Section \ref{AfdGepunt}, but now of course using%
\begin{equation}
\mathcal{N}_{J,j}:=\{\eta-\theta:\eta,\theta\in\mathcal{L}^{J}\text{ such that
}\eta_{\widehat{j}}=\theta_{\widehat{j}}\}\text{ \ for \ }j\in J \label{N_Jj}%
\end{equation}
instead of $\mathcal{N}_{1},...,\mathcal{N}_{n}$, where the latter led to the
$W_{1}$ seminorm on $\mathcal{O}_{\mathcal{L}}$. Here
\[
\eta_{\widehat{j}}:=\eta^{J\backslash\{j\}}%
\]
for $\eta\in\mathcal{L}^{J}$, which can equivalently be defined by
(\ref{jKomplRed}), but using $A_{i}$, $B_{i}$ and $\nu_{i}$ only for $i\in J$
when setting up Section \ref{AfdGepunt}, rather than for the entire range
$i=1,...,n$. The semi unit ball leading to this $W_{1}$ seminorm is
\begin{equation}
\mathcal{C}_{\mathcal{L}^{J}}:=\operatorname*{conv}\mathcal{N}_{J},
\label{substrukC}%
\end{equation}
where
\begin{equation}
\mathcal{N}_{J}:=\bigcup_{j\in J}\mathcal{N}_{J,j}. \label{substrukN}%
\end{equation}

\subsection{Reducible pointed $W_{1}$ structures\label{OndAfdRedStruk}}

A natural question is whether the semi unit ball $\mathcal{C}_{\mathcal{L}%
^{J}}$ above can be obtained as the reduction $\mathcal{C}_{\mathcal{L}}^{J}$
of the original semi unit ball $\mathcal{C}_{\mathcal{L}}$. Similarly for
$\mathcal{O}_{\mathcal{L}^{J}}$ and the sets $\mathcal{N}_{J,j}$. These
questions will in fact become relevant in the next subsection, when we reach
the main goal of this section, namely to prove additivity properties.

In order to answer these questions positively, assumptions beyond those in
Definition \ref{gepunteStruk} need to be made. We note that these assumptions
will automatically be satisfied in the C*-algebraic framework.

As in the previous subsection we consider a pointed $(W_{1},n)$ structure
$(A,B,\nu,\mathcal{L})$ and any $J\in\mathcal{P}_{n}.$

To avoid any mismatches and ambiguities, we always need to preserve the
ordering of the $A_{j}$'s in any tensor product of them. Similarly for the
$B_{j}$'s. Therefore the notation%
\[
A_{I}\vec{\odot}A_{J}:=A_{I\cup J}\text{ \ and \ \ }B_{I}\vec{\otimes}%
B_{J}:=B_{I\cup J}%
\]
will be used for any $I,J\in\mathcal{P}_{n}$ with no points in common, i.e.,
$I\cap J=\varnothing$. But then the tensor product of maps $\eta\in
L(A_{I},B_{I})$ and $\theta\in L(A_{J},B_{J})$ for $I,J\in\mathcal{P}_{n}$
with $I\cap J=\varnothing$ need to be defined correspondingly. This is indeed
possible. Note that transpositions of adjacent $A_{j}$'s in any tensor product
of $A_{j}$'s are linear bijections, compositions of which in particular give
us a natural unique well defined \emph{ordering map}%
\[
\alpha_{IJ}:A_{I}\odot A_{J}\rightarrow A_{I}\vec{\odot}A_{J},
\]
such that%
\[
\alpha_{IJ}((\otimes_{i\in I}a_{i})\otimes(\otimes_{j\in J}a_{j}%
))=\otimes_{j\in I\cup J}a_{j}%
\]
for arbitrary $a_{j}\in A_{j}$. As a simple example to clarify the meaning of
this, suppose $I=\{1,3,5\}$ and $J=\{2,4\}$, then for $a_{j}\in A_{j}$, we
have $\alpha_{IJ}(a_{1}\otimes a_{3}\otimes a_{5}\otimes a_{2}\otimes
a_{4})=a_{1}\otimes a_{2}\otimes a_{3}\otimes a_{4}\otimes a_{5}$. Note that
$\alpha_{IJ}$ can be viewed as a \emph{pointed space isomorphism}, i.e., a
bijection $\alpha_{IJ}\in L_{u}(A_{I}\odot A_{J},A_{I}\vec{\odot}A_{J})$.

Similarly we have the ordering map
\[
\beta_{IJ}^{0}:B_{I}\odot B_{J}\rightarrow B_{I}\vec{\odot}B_{J}.
\]
Clearly $\beta_{IJ}^{0}\in L_{u}(B_{I}\odot B_{J},B_{I}\vec{\odot}B_{J})$.

However, if the tensor products $B_{J}$ for $J\subset\lbrack n]$ are indeed
completed in some norm, we need to \emph{assume} that $\beta_{IJ}^{0}$ is
continuous in this norm, and therefore extends uniquely to a continuous
bijection%
\[
\beta_{IJ}\in L_{u}(B_{I}\otimes B_{J},B_{I}\vec{\otimes}B_{J}),
\]
which is the uniquely defined pointed space isomorphism serving as the
ordering map on $B_{I}\otimes B_{J}$.

Given this, we can define%
\[
\eta\vec{\odot}\theta:A_{I}\vec{\odot}A_{J}\rightarrow B_{I}\vec{\otimes}B_{J}%
\]
as%
\[
\eta\vec{\odot}\theta:=\beta_{IJ}\circ(\eta\odot\theta)\circ\alpha_{IJ}^{-1}.
\]
It has the following expected property.

\begin{proposition}
In terms of the notation and assumptions so far in this subsection,
\[
\eta\vec{\odot}\theta=\theta\vec{\odot}\eta,
\]
for all $\eta\in L(A_{I},B_{I})$ and $\theta\in L(A_{J},B_{J})$, where
$I,J\in\mathcal{P}_{n}$ with $I\cap J=\varnothing$.
\end{proposition}

\begin{proof}
Note that by the definition of the ordering maps, we have $\beta_{IJ}(c\otimes
d)=\beta_{JI}(d\otimes c)$ for elements in the algebraic tensor products,
$c\in\odot_{i\in I}B_{i}$ and $d\in\odot_{j\in J}B_{j}$. Now, for arbitrary
$a_{j}\in A_{j}$,
\begin{align*}
(\eta\odot\theta)\circ\alpha_{IJ}^{-1}(\otimes_{j\in I\cup J}a_{j})  &
=\eta(\otimes_{i\in I}a_{i})\otimes\theta(\otimes_{j\in J}a_{j})\\
(\theta\odot\eta)\circ\alpha_{JI}^{-1}(\otimes_{j\in I\cup J}a_{j})  &
=\theta(\otimes_{j\in J}a_{j})\otimes\eta(\otimes_{i\in I}a_{i})
\end{align*}
Approximate $\eta(\otimes_{i\in I}a_{i})$ and $\theta(\otimes_{j\in J}a_{j})$
by sequences $(c_{l})$ and $(d_{l})$ in the algebraic tensor products
$\odot_{i\in I}B_{i}$ and $\odot_{j\in J}B_{j}$ respectively. Since
$\beta_{IJ}(c_{l}\otimes d_{l})=\beta_{JI}(d_{l}\otimes c_{l})$ and
$\beta_{IJ}$ and $\beta_{JI}$ are assumed to be continuous, it follows that%
\[
\beta_{IJ}\circ(\eta\odot\theta)\circ\alpha_{IJ}^{-1}(\otimes_{j\in I\cup
J}a_{j})=\beta_{JI}\circ(\theta\odot\eta)\circ\alpha_{JI}^{-1}(\otimes_{j\in
I\cup J}a_{j}),
\]
as required.
\end{proof}

We also need to strengthen (\ref{tensProdBehoud}) in Lemma \ref{teleskoop} to
the following: \emph{Assume} that
\begin{equation}
\eta^{J}\vec{\odot}\theta^{J^{\prime}}\in\mathcal{L} \label{uitbrVanL^J}%
\end{equation}
for all $\eta,\theta\in\mathcal{L}$ and $J\in\mathcal{P}_{n}$.

\begin{definition}
\label{redGepuntStruk}The pointed $(W_{1},n)$ structure $(A,B,\nu
,\mathcal{L})$ is called \emph{reducible} if both the above mentioned
assumptions are indeed satisfied, namely (\ref{uitbrVanL^J}) and the existence
of the continuous ordering maps $\beta_{IJ}$ in the case of completed $B$.
\end{definition}

Note that for a non-empty $I\subset J^{\prime}$, it follows from
(\ref{uitbrVanL^J}) and the definition of reduction in the previous
subsection, that
\begin{equation}
\eta^{J}\vec{\odot}\theta^{I}=(\eta^{J}\vec{\odot}\theta^{J^{\prime}})^{J\cup
I}\in\mathcal{L}^{J\cup I}. \label{algUitbrVanL^J}%
\end{equation}
In particular, this gives the following simple result.

\begin{proposition}
\label{redSubstr}If the pointed $(W_{1},n)$ structure $(A,B,\nu,\mathcal{L})$
is reducible, then so is its pointed $W_{1}$ substructures.
\end{proposition}

The reason for the terminology ``reducible'' in Definition
\ref{redGepuntStruk}, is that the semi unit ball of $(A_{J},B_{J},\nu
_{J},\mathcal{L}^{J})$ is then obtained from that of $(A,B,\nu,\mathcal{L})$
by reduction. This and related facts are shown below.

In terms of the setup and notation of this section, we have the following.

\begin{lemma}
\label{N_J,jUitbrei}Assume that $(A,B,\nu,\mathcal{L})$ is a reducible pointed
$(W_{1},n)$ structure and consider any $J\in\mathcal{P}_{n}$. For all
$\lambda\in\mathcal{N}_{J,j}$ with $j\in J$, and $\zeta\in\mathcal{L}%
^{J^{\prime}}$, it follows that $\lambda\vec{\odot}\zeta\in\mathcal{N}_{j}$.
\end{lemma}

\begin{proof}
Note that $(\lambda\vec{\odot}\zeta)_{\widehat{j}}=\lambda_{\widehat{j}}%
\vec{\odot}\zeta=0$. In terms of $\lambda=\eta-\theta$ with $\eta,\theta
\in\mathcal{L}^{J}$, this means that $\lambda\vec{\odot}\zeta=\eta\vec{\odot
}\zeta-\theta\vec{\odot}\zeta$, where $\eta\vec{\odot}\zeta,\theta\vec{\odot
}\zeta\in\mathcal{L}$ by (\ref{uitbrVanL^J}) and $(\eta\vec{\odot}%
\zeta)_{\widehat{j}}=(\theta\vec{\odot}\zeta)_{\widehat{j}}$, as needed in
(\ref{N_j}).
\end{proof}

\begin{proposition}
\label{redNj}Assume that $(A,B,\nu,\mathcal{L})$ is a reducible pointed
$(W_{1},n)$ structure and consider any $J\in\mathcal{P}_{n}$. Then
$\mathcal{N}_{j}^{J}=\{0\}$ when $j\in J^{\prime}$, while
\[
\mathcal{N}_{j}^{J}=\mathcal{N}_{J,j}%
\]
for $j\in J$.
\end{proposition}

\begin{proof}
By (\ref{S^J}), $\mathcal{N}_{j}^{J}:=\{\eta^{J}:\eta\in\mathcal{N}_{j}\}$.
For $j\in J^{\prime}$ the reduction over $j$ is included in the reduction over
$J^{\prime}$, hence $\lambda_{\widehat{j}}=0$ for $\lambda\in L(A,B)$ implies
that $\lambda^{J}=0$, directly from the definitions of $\lambda_{\widehat{j}}$
and $\lambda^{J}$. Thus $\mathcal{N}_{j}^{J}=\{0\}$ when $j\in J^{\prime}$.
Now assume that $j\in J$. Consider any $\lambda^{J}\in\mathcal{N}_{j}^{J}$,
i.e., we take $\lambda=\eta-\theta$ with $\eta,\theta\in\mathcal{L}$ and
$\eta_{\widehat{j}}=\theta_{\widehat{j}}$. Since $j\in J$, both sides of
$(\eta^{J})_{\widehat{j}}=(\eta_{\widehat{j}})^{J}$ are well defined, and
indeed equal by the definitions of these reductions. Similarly for $\theta$,
which means that $(\eta^{J})_{\widehat{j}}=(\theta^{J})_{\widehat{j}}$, hence
$\lambda^{J}\in\mathcal{N}_{J,j}$ by (\ref{N_Jj}). This shows that
$\mathcal{N}_{j}^{J}\subset\mathcal{N}_{J,j}$, even if $(A,B,\nu,\mathcal{L})$
is not assumed reducible. Conversely, consider any $\lambda\in\mathcal{N}%
_{J,j}$. For any $\zeta\in\mathcal{L}^{J^{\prime}}$ it then follows from Lemma
\ref{N_J,jUitbrei} that $\lambda\vec{\odot}\zeta\in\mathcal{N}_{j}$.
Consequently, $\lambda=(\lambda\vec{\odot}\zeta)^{J}\in\mathcal{N}_{j}^{J}$,
proving that $\mathcal{N}_{J,j}\subset\mathcal{N}_{j}^{J}$.
\end{proof}

In particular this tells us that the reductions $\mathcal{N}_{j}^{J}$ of
$\mathcal{N}_{j}$ for $j\in J$, play the same role for $(A_{J},B_{J},\nu
_{J},\mathcal{L}^{J})$ as $\mathcal{N}_{1},...,\mathcal{N}_{n}$ play for the
reducible pointed $W_{1}$ structure $(A,B,\nu,\mathcal{L})$.

\begin{corollary}
\label{redONC}In Proposition \ref{redNj} we have
\[
\mathcal{O}_{\mathcal{L}^{J}}=\mathcal{O}_{\mathcal{L}}^{J}\text{,
\ }\mathcal{N}_{J}=\mathcal{N}^{J}\text{ \ and \ }\mathcal{C}_{\mathcal{L}%
^{J}}=\mathcal{C}_{\mathcal{L}}^{J}%
\]
for (\ref{O_LJ}), (\ref{substrukN}) and (\ref{substrukC}).
\end{corollary}

\begin{proof}
From Proposition \ref{redNj}, (\ref{bal}) and (\ref{substrukN}) one has
$\mathcal{N}_{J}=\mathcal{N}^{J}$, hence
\[
\mathcal{C}_{\mathcal{L}^{J}}=\operatorname*{conv}\mathcal{N}^{J}%
=(\operatorname*{conv}\mathcal{N})^{J}=\mathcal{C}_{\mathcal{L}}^{J}%
\]
and%
\[
\mathcal{O}_{\mathcal{L}^{J}}=\operatorname*{span}\nolimits_{\mathbb{R}%
}\mathcal{N}^{J}=(\operatorname*{span}\nolimits_{\mathbb{R}}\mathcal{N}%
)^{J}=\mathcal{O}_{\mathcal{L}}^{J},
\]
because of Lemma \ref{teleskoop}.
\end{proof}

As one may expect, $\mathcal{O}_{\mathcal{L}^{J}}=\mathcal{O}_{\mathcal{L}%
}^{J}$ can alternatively be proved along the lines of the proof of Proposition
\ref{redNj}.

We can also use Lemma \ref{N_J,jUitbrei} along with Corollary \ref{redONC} to
obtain the following property of $\left\|  \cdot\right\|  _{W_{1}}$.

\begin{proposition}
\label{normbehoud}Assume that $(A,B,\nu,\mathcal{L})$ is a reducible pointed
$(W_{1},n)$ structure and consider any $I,J\in\mathcal{P}_{n}$ with $I\cap
J=\varnothing$. Then
\[
\left\|  (\eta-\theta)\vec{\odot}\zeta\right\|  _{W_{1}}=\left\|  \eta
-\theta\right\|  _{W_{1}}%
\]
for all $\eta,\theta\in\mathcal{L}^{I}$ and $\zeta\in\mathcal{L}^{J}$.
\end{proposition}

\begin{proof}
For any $\eta,\theta\in\mathcal{L}^{I}$ and $\zeta\in\mathcal{L}^{J}$, set
$\lambda:=(\eta-\theta)\vec{\odot}\zeta\in\mathcal{L}^{I\cup J}$, according to
(\ref{algUitbrVanL^J}), then $\lambda^{I}=\eta-\theta$.

For any $\gamma\in\mathcal{C}_{\mathcal{L}^{I\cup J}}$ and $t\geq0$ such that
$\lambda=t\gamma$, one has that $\lambda^{I}=t\gamma^{I}$. By Corollary
\ref{redONC}, but applied to the reducible (Proposition \ref{redSubstr})
pointed $W_{1}$ structure $(A_{I\cup J},B_{I\cup J},\nu_{I\cup J}%
,\mathcal{L}^{I\cup J})$, we know that $\gamma^{I}\in\mathcal{C}%
_{\mathcal{L}^{I}}$, from which it follows that $\left\|  \lambda^{I}\right\|
_{W_{1}}\leq t$, hence $\left\|  \lambda^{I}\right\|  _{W_{1}}\leq\left\|
\lambda\right\|  _{W_{1}}$ by the definition of $\left\|  \cdot\right\|
_{W_{1}}$ in Corollary \ref{W-seminorm} (via Proposition \ref{normUitBal}).

Conversely, consider any $\gamma\in\mathcal{C}_{\mathcal{L}^{I}}$ and $t\geq0$
such that $\lambda^{I}=t\gamma$. By (\ref{substrukC}) we have
\[
\gamma=\sum_{i=1}^{l}p_{i}\gamma_{i}%
\]
for some $p_{1},...,p_{l}>0$ with $p_{1}+...+p_{l}=1$, and $\gamma_{i}%
\in\mathcal{N}_{I,j_{i}}$ for some $j_{i}\in I$. Because of Lemma
\ref{N_J,jUitbrei} applied to $(A_{I\cup J},B_{I\cup J},\nu_{I\cup
J},\mathcal{L}^{I\cup J})$, it follows that
\[
\gamma\vec{\odot}\zeta=\sum_{i=1}^{l}p_{i}\gamma_{i}\vec{\odot}\zeta
\in\mathcal{C}_{\mathcal{L}^{I\cup J}}.
\]
Since $\lambda=t\gamma\vec{\odot}\zeta$, we conclude that $\left\|
\lambda\right\|  _{W_{1}}\leq t$, thus $\left\|  \lambda\right\|  _{W_{1}}%
\leq\left\|  \lambda^{I}\right\|  _{W_{1}}$.
\end{proof}

These results will be applied in the next subsection to prove the additivity
of $W_{1}$ with respect to tensor products.

\subsection{Additivity\label{OndAfdAd}}

We now arrive at this section's main results in the abstract pointed space
setup, which will be applied to the C*-algebras in the next subsection.
Consider a reducible pointed $(W_{1},n)$ structure $(A,B,\nu,\mathcal{L})$ and
any $m$\emph{-partition} $P$ of $[n]$, by which we mean a function
$P:[m]\rightarrow$ $\mathcal{P}_{n}$ such that $P(1)\cup...\cup P(m)=[n]$ and
$P(k)\cap P(l)=\varnothing$ for $k\neq l$. Our goal is to determine how
$\left\|  \cdot\right\|  _{W_{1}}$ for $(A,B,\nu,\mathcal{L})$ relates to
$\left\|  \cdot\right\|  _{W_{1}}$ for the $(A_{P(k)},B_{P(k)},\nu
_{P(k)},\mathcal{L}^{P(k)})$'s via reduction. Similarly for $W_{1}$, but
specifically for product maps. Using the results of the previous subsection,
these relationships can be stated as a form of ``reductive superadditivity''
of $\left\|  \cdot\right\|  _{W_{1}}$ and an additivity property of $W_{1}$.

\begin{theorem}
\label{gepunteAdEiensk}Let $(A,B,\nu,\mathcal{L})$ be a reducible pointed
$(W_{1},n)$ structure and $P$ any $m$-partition of $[n]$. Then%
\begin{equation}
\left\|  \lambda\right\|  _{W_{1}}\geq\sum_{k=1}^{m}\left\|  \lambda
^{P(k)}\right\|  _{W_{1}} \label{gepuntSupAd}%
\end{equation}
for all $\lambda\in\mathcal{O}_{\mathcal{L}}$, while
\begin{equation}
W_{1}(\eta_{1}\vec{\odot}...\vec{\odot}\eta_{m},\theta_{1}\vec{\odot}%
...\vec{\odot}\theta_{m})=\sum_{k=1}^{m}W_{1}(\eta_{k},\theta_{k})
\label{gepunteAd}%
\end{equation}
for all $\eta_{k},\theta_{k}\in\mathcal{L}^{P(k)}$, for $k=1,...,m$.
\end{theorem}

\begin{proof}
Let $\lambda\in\mathcal{O}_{\mathcal{L}}$. Consider any $\gamma\in
\mathcal{C}_{\mathcal{L}}$ and $t\geq0$ such that $\lambda=t\gamma$. By
$\mathcal{C}_{\mathcal{L}}=\operatorname*{conv}\mathcal{N}$,
\[
\gamma=\sum_{i=1}^{l}p_{i}\alpha_{i}%
\]
for some $p_{1},...,p_{l}>0$ with $p_{1}+...+p_{l}=1$, and $\alpha_{i}%
\in\mathcal{N}=\cup_{j\in\lbrack n]}\mathcal{N}_{j}$. Set
\[
q_{k}:=\sum_{i\in R_{k}}p_{i}\text{ \ and \ }\gamma_{k}:=\sum_{i\in R_{k}%
}\frac{p_{i}}{q_{k}}\alpha_{i}\in\mathcal{C}_{\mathcal{L}},
\]
where $R_{k}:=\{i\in\lbrack l]:\alpha_{i}\in\cup_{j\in P(k)}\mathcal{N}%
_{j}\}\backslash(R_{1}\cup...\cup R_{k-1})$ for $k=1,...,m$, with $R_{1}%
\cup...\cup R_{k-1}=\varnothing$ for $k=1$. Then by Proposition \ref{redNj},%
\[
\gamma=\sum_{k=1}^{m}q_{k}\gamma_{k}\text{ \ and \ }\gamma^{P(k)}=q_{k}%
\gamma_{k}^{P(k)}%
\]
with $\gamma_{k}^{P(k)}\in\mathcal{C}_{\mathcal{L}}^{P(k)}=\mathcal{C}%
_{\mathcal{L}^{P(k)}}$ by Corollary \ref{redONC}. This tells us that
$\lambda^{P(k)}=t\gamma^{P(k)}=tq_{k}\gamma_{k}^{P(k)}\in tq_{k}%
\mathcal{C}_{\mathcal{L}^{P(k)}}$, hence $\left\|  \lambda^{P(k)}\right\|
_{W_{1}}\leq tq_{k}$, thus
\[
\sum_{k=1}^{m}\left\|  \lambda^{P(k)}\right\|  _{W_{1}}\leq t\text{,}%
\]
implying (\ref{gepuntSupAd}), by $\left\|  \cdot\right\|  _{W_{1}}$'s
definition in Corollary \ref{W-seminorm}.

In particular, for $\eta,\theta\in\mathcal{L}$ we have
\[
\left\|  \eta-\theta\right\|  _{W_{1}}\geq\sum_{k=1}^{m}\left\|  \eta
^{P(k)}-\theta^{P(k)}\right\|  _{W_{1}}.
\]
On the other hand, given $\eta_{k},\theta_{k}\in\mathcal{L}^{P(k)}$, by
Proposition \ref{normbehoud} we have%
\begin{align*}
&  \left\|  \eta_{1}\vec{\odot}...\vec{\odot}\eta_{m}-\theta_{1}\vec{\odot
}...\vec{\odot}\theta_{m}\right\|  _{W_{1}}\\
&  \leq\left\|  (\eta_{1}-\theta_{1})\vec{\odot}\eta_{2}\vec{\odot}%
...\vec{\odot}\eta_{m}\right\|  _{W_{1}}+\left\|  \theta_{1}\vec{\odot}%
(\eta_{2}\vec{\odot}...\vec{\odot}\eta_{m}-\theta_{2}\vec{\odot}...\vec{\odot
}\theta_{m})\right\|  _{W_{1}}\\
&  =\left\|  \eta_{1}-\theta_{1}\right\|  _{W_{1}}+\left\|  \eta_{2}\vec
{\odot}...\vec{\odot}\eta_{m}-\theta_{2}\vec{\odot}...\vec{\odot}\theta
_{m}\right\|  _{W_{1}}\\
&  \vdots\\
&  \leq\sum_{k=1}^{m}\left\|  \eta_{k}-\theta_{k}\right\|  _{W_{1}}.
\end{align*}
These two inequalities prove (\ref{gepunteAd}).
\end{proof}

With this result in hand, we can turn to the final aim of this section, namely
additivity in the C*-algebraic framework.

\subsection{C*-algebras\label{OndAfdC*ad}}

The question being studied at the moment, is how $W_{1}$ between product
channels from one compound system to another, relate to $W_{1}$ between the
channels composing the product channels. This was answered in a more abstract
form in Theorem \ref{gepunteAdEiensk} of the previous subsection, namely they
are related in a simple additive way. Now we essentially just translate this
additivity of $W_{1}$ to the C*-algebraic framework of Section \ref{AfdC*}.

Given an $n$-composite C*-system $(A_{j},B_{j},\nu_{j}:j=1,...,n)$ as in
Definition \ref{n-C*-stelsel}, let $[n]:=\{1,...,n\}$, let $\mathcal{P}_{n}$
be the collection of non-empty proper subsets of $[n]$, and define an
$m$\emph{-partition} $P$ of $[n]$, as a function $P:[m]\rightarrow$
$\mathcal{P}_{n}$ such that $P(1)\cup...\cup P(m)=[n]$ and $P(k)\cap
P(l)=\varnothing$ for $k\neq l$. Furthermore, we set
\[
A_{J}:=\bigotimes_{j\in J}A_{j}\text{ \ and \ }B_{J}:=\bigotimes_{j\in J}B_{j}%
\]
for any $J\in\mathcal{P}_{n}$, both being minimal tensor products or both
maximal tensor products. As in Section \ref{AfdC*}, either all tensor products
here are minimal, or all are maximal. For any $I,J\in\mathcal{P}_{n}$ with no
points in common, i.e., $I\cap J=\varnothing$, consider the ordered tensor
products%
\[
A_{I}\vec{\otimes}A_{J}:=A_{I\cup J}\text{ \ and \ \ }B_{I}\vec{\otimes}%
B_{J}:=B_{I\cup J}%
\]
and the corresponding ordered tensor product%
\[
\eta\vec{\otimes}\theta:A_{I}\vec{\otimes}A_{J}\rightarrow B_{I}\vec{\otimes
}B_{J}%
\]
of bounded linear maps $\eta:A_{I}\rightarrow B_{I}$ and $\theta
:A_{J}\rightarrow B_{J}$, defined as in Subsection \ref{OndAfdRedStruk}, but
via the minimal or maximal tensor product $\eta\otimes\theta$, instead of the
algebraic tensor product, and where the ordering map $\alpha_{IJ}$ is of
course extended to the completed tensor product, like $\beta_{IJ}$ is. In the
C*-algebraic case, this extension of the ordering maps is automatically
possible (see the proof of Proposition \ref{kanRed&Uitbr} below). Note that
the reduction $\eta^{J}$ of $\eta\in L(A,B)$ is defined analogously to
Subsection \ref{OndAfdGepuntSubstr}, the only difference being that $A$ and
$A_{J}$ are now completed tensor products.

The next result is the key to convert Theorem \ref{gepunteAdEiensk} to the
C*-algebraic setup.

\begin{proposition}
\label{kanRed&Uitbr}Consider the setup of this subsection. For any
$J\in\mathcal{P}_{n}$, it follows that
\[
K(A,B)^{J}=K(A_{J},B_{J}).
\]
In addition, for any $D\in K(A_{I},B_{I})$ and $E\in K(A_{J},B_{J})$ with
$I,J\in\mathcal{P}_{n}$ such that $I\cap J=\varnothing$, one has
\[
D\vec{\otimes}E\in K(A_{I}\vec{\otimes}A_{J},B_{I}\vec{\otimes}B_{J}).
\]
\end{proposition}

\begin{proof}
Starting with the latter statement, note that in this C*-algebraic setup the
ordering maps from Subsection \ref{OndAfdRedStruk} are $\ast$-isomorphisms
$\alpha_{IJ}:A_{I}\otimes A_{J}\rightarrow A_{I}\vec{\otimes}A_{J}$ and
$\beta_{IJ}:B_{I}\otimes B_{J}\rightarrow B_{I}\vec{\otimes}B_{J}$ (and
therefore extended to the completions and u.c.p.) from the outset, since they
are compositions of the transposition maps mentioned in Subsection
\ref{OndAfdRedStruk}, which are indeed $\ast$-isomorphisms (see for example
\cite[II.9.2.6]{Bl}). Hence $D\vec{\otimes}E:=\beta_{IJ}\circ(D\otimes
E)\circ\alpha_{IJ}^{-1}$ is a composition of u.c.p. maps, since $D\vec
{\otimes}E$ is a channel. Consequently $D\vec{\otimes}E\in K(A_{I}\vec
{\otimes}A_{J},B_{I}\vec{\otimes}B_{J})$, as required.

In particular, for $I=J^{\prime}$, it follows that $E=(D\vec{\otimes}E)^{J}\in
K(A,B)^{J}$. Hence $K(A_{J},B_{J})\subset K(A,B)^{J}$. Conversely, since the
reduction of a channel is again a channel, we have $K(A,B)^{J}\subset
K(A_{J},B_{J})$.
\end{proof}

Now we can answer the above mentioned question as follows, along with the
``reductive superadditivity'' of $\left\|  \cdot\right\|  _{W_{1}}$.

\begin{theorem}
\label{C*AdEiensk}Let $(A_{j},B_{j},\nu_{j}:j=1,...,n)$ be an $n$-composite
C*-system and $P$ any $m$-partition of $[n]$. Then in both the minimal and the
maximal tensor product setup,
\[
W_{1}(D_{1}\vec{\otimes}...\vec{\otimes}D_{m},E_{1}\vec{\otimes}%
...\vec{\otimes}E_{m})=\sum_{k=1}^{m}W_{1}(D_{k},E_{k})
\]
for all $D_{k},E_{k}\in K(A_{P(k)},B_{P(k)})$, for $k=1,...,m$. In addition,
\[
\left\|  \lambda\right\|  _{W_{1}}\geq\sum_{k=1}^{m}\left\|  \lambda
^{P(k)}\right\|  _{W_{1}}%
\]
for all $\lambda\in\mathcal{O}_{A,B}$ as defined in Section \ref{AfdC*}.
\end{theorem}

\begin{proof}
If the algebraic $(W_{1},n)$ structure $(A_{\odot},B,\nu,K(A,B)|_{\odot})$,
obtained as in the proof of Theorem \ref{C*-W-norm} from $(A_{j},B_{j},\nu
_{j}:j=1,...,n)$, is a reducible pointed $(W_{1},n)$ structure, and%
\begin{equation}
\left(  K(A,B)|_{\odot}\right)  ^{J}=K(A_{J},B_{J})|_{\odot},
\label{KanStelRed}%
\end{equation}
for $J\in\mathcal{P}_{n}$, then the theorem follows directly from Theorem
\ref{gepunteAdEiensk}. This is because of the continuity of all the maps
involved and the resulting one-to-one correspondence between $K(A,B)|_{\odot}$
and $K(A,B)$ explained in the proof of Theorem \ref{C*-W-norm}, and similarly
for $K(A_{J},B_{J})|_{\odot}$ and $K(A_{J},B_{J})$, which allows us to
translate directly between the algebraic and completed tensor products.
Condition (\ref{KanStelRed}) is needed to ensure that for $\mathcal{L}%
=K(A,B)|_{\odot}$ in Theorem \ref{gepunteAdEiensk}, we have $\mathcal{L}%
^{P(k)}=K(A_{P(k)},B_{P(k)})|_{\odot}$. Note that (\ref{KanStelRed}) indeed
holds because of Proposition \ref{kanRed&Uitbr} and
$\left(  K(A,B)|_{\odot}\right)  ^{J}=K(A,B)^{J}|_{\odot}$, the latter being
true since the only difference between the reductions on the two sides is that
they respectively involve restrictions to the algebraic and a completed tensor
product of $A_{j}$'s. We are simply left to verify reducibility.

As mentioned in the proof of Proposition \ref{kanRed&Uitbr}, the $\beta_{IJ}%
$'s are $\ast$-isomorphisms and therefore continuous, verifying the one
condition for reducibility in Definition \ref{redGepuntStruk}. The other
condition is guaranteed by the special case of Proposition \ref{kanRed&Uitbr}
with $I\cup J=[n]$.
\end{proof}

Recall that the special case in finite dimensions of the first part of this
theorem was already mentioned in Section \ref{AfdBuitL}, in a special form
where the ordering of the tensor products was unnecessary.

As an immediate consequence of this theorem, we obtain the following stability
result for $W_{1}$.

\begin{corollary}
\label{stab}In Theorem \ref{C*AdEiensk}, let $P$ a $2$-partition of $[n]$.
Then
\[
W_{1}(D_{1}\vec{\otimes}F_{2},E_{1}\vec{\otimes}F_{2})=W_{1}(D_{1},E_{1})
\]
for all $D_{1},E_{1}\in K(A_{P(1)},B_{P(1)})$ and any $F_{2}\in K(A_{P(2)}%
,B_{P(2)})$. Similarly,
\[
W_{1}(F_{1}\vec{\otimes}D_{2},F_{1}\vec{\otimes}E_{2})=W_{1}(D_{2},E_{2})
\]
for all $D_{2},E_{2}\in K(A_{P(2)},B_{P(2)})$ and any $F_{1}\in K(A_{P(1)}%
,B_{P(1)})$.
\end{corollary}

\begin{proof}
This is simply because $W_{1}(F_{2},F_{2})=0$ in the first case. Similarly for
the second.
\end{proof}

Stability for other distances between channels have been presented and
discussed in for example \cite{AKitN} and \cite{GLN}. The latter in particular
emphasizes the utility of stability. These references only treated the case
where the $F_{k}$'s in the corollary above were identity maps.

Of course, stability for pointed $W_{1}$ structures similarly follow from
Theorem \ref{gepunteAdEiensk}.

\section{Outlook\label{AfdEinde}}

In this paper our focus has been the mathematical development of a Wasserstein
distance of order 1 between channels from one composite system to another,
denoted by $W_{1}$. We have not yet investigated any relations or comparisons
of $W_{1}$ with the diamond norm (mentioned in the introduction) or other
distances between channels. See \cite{GLN} for a critical assessment of
various such distances. A logical first step is to do this in finite
dimensions on matrix algebras $M_{d}$, where one has a simple canonical trace
which plays an important role in the diamond (aka, completely bounded trace)
norm. Indeed, \cite{DMTL} extensively explored the relation between $W_{1}$
for states and the trace norm, including the characterization of $W_{1}$ for
the case of $n=1$, i.e., for single systems rather than composed systems. This
can analogously be explored for the case of channels.

We have also not yet explored applications of this distance. Applications of
Wasserstein distance of order 1 between states in the finite dimensional case
developed in \cite{DMTL}, have already been treated and proposed in
\cite{DMTL} itself, as well as in a number of papers \cite{DMRS, DR, HRS, K,
L, RS} in various contexts. We expect that the approach of this paper to the
case of channels should be similarly applicable. In addition, our abstract
approach has the potential to allow for applications in other contexts than
quantum channels and for further mathematical development.

At the end of Subsection \ref{OndAfdC*ad} we pointed out that $W_{1}$
satisfies stability. However, there are other properties that one may also
want a distance between channels to satisfy, depending on the application. See
for example \cite{GLN}. Since $W_{1}$ is a metric, a number of basic
properties are already satisfied. One property that we have not discussed in
this paper is chaining or bounds on the overall difference between composed
channels; see \cite{GLN} and \cite[Subsection 5.4]{AKitN}. Whether or not this
or similar and other properties hold for $W_{1}$, certainly warrants further investigation.

\end{document}